\documentclass[a4paper, 10pt]{article}
\usepackage[margin=1in]{geometry}
\usepackage{srcltx,graphicx,epstopdf}
\usepackage{amsmath, amssymb, amsthm}
\usepackage{color}
\usepackage{multirow}
\usepackage{overpic}
\usepackage{bm}
\usepackage{subfig} 
\usepackage{hyperref}
\usepackage{booktabs}
\usepackage{tikz}

\usepackage{cleveref}
\crefname{section}{Section}{Sections}
\crefname{subsection}{Subsection}{Subsections}
\crefname{appendix}{Appendix}{Appendix}
\crefname{figure}{Figure}{Figures}
\crefname{table}{Table}{Tables}
\crefname{property}{Property}{Properties}
\crefname{theorem}{Theorem}{Theorem}
\crefname{criterion}{criterion}{criteria}

\graphicspath{{images/}}

\newtheorem{theorem}{Theorem}
\newtheorem{lemma}[theorem]{Lemma}

\newtheorem{property}[theorem]{Property}
\newtheorem{corollary}[theorem]{Corollary}

\newtheorem{remark}{Remark}

\newcommand\bA{{\bf A}}

\newcommand\bD{{\bf D}}

\newcommand\bM{{\bf M}}

\newcommand\bR{{\bf R}}

\newcommand\bG{{\bf G}}
\newcommand\bLambda{{\bf \Lambda}}

\newcommand\bbR{\mathbb{R}}
\newcommand\bbN{\mathbb{N}}

\newcommand\bbH{\mathbb{H}}

\newcommand\bbS{\mathbb{S}}
\newcommand\bg{\boldsymbol{g}}
\newcommand\be{\boldsymbol{e}}
\newcommand\ba{\boldsymbol{a}}
\newcommand\bb{\boldsymbol{b}}
\newcommand\bp{\boldsymbol{p}}
\newcommand\bc{\boldsymbol{c}}

\newcommand\bw{\boldsymbol{w}}
\newcommand\bx{\boldsymbol{x}}
\newcommand\bF{{\boldsymbol{F}}}

\newcommand\bU{{\boldsymbol{U}}}
\newcommand\bOmega{\bm{\Omega}}
\newcommand\bxi{\bm{\xi}}

\newcommand\dd{\,\mathrm{d}}

\newcommand\PN{{$P_N$ }}
\newcommand\MN{{$M_N$ }}
\newcommand\SN{{$S_N$ }}
\newcommand\MP[1]{{${M\!P}_{#1}$ }}
\newcommand\MPN{{\MP{\!N}}}
\newcommand\Mone{{$M_1$ }}

\newcommand\HMP[1]{{${H\!M\!P}_{#1}$ }}
\newcommand\HMPN{{\HMP{\!N}}}

\newcommand\mE{{\mathcal{E}}}

\newcommand\mS{{\mathcal{S}}}

\newcommand\mP{\mathcal{P}}
\newcommand\mPn{\tilde{\mathcal{P}}}
\newcommand\mK{{\mathcal{K}}}
\newcommand\mI{{\mathcal{I}}}
\newcommand\mKn{{\tilde{\mK}}}

\newcommand\weight{\omega^{[\bc_0]}}
\newcommand\weightoned{\omega^{[c_0]}}
\newcommand\pl{\phi^{[\bc_0]}}
\newcommand\ploned{\phi^{[c_0]}}
\newcommand\Pl{\Phi^{[\bc_0]}}
\newcommand\Ploned{\Phi^{[c_0]}}

\newcommand\weightn{\tilde{\omega}^{[\bc_0]}}
\newcommand\weightnoned{\tilde{\omega}^{[c_0]}}
\newcommand\Pln{\tilde{\Phi}^{[\bc_0]}}
\newcommand\Plnoned{\tilde{\Phi}^{[c_0]}}
\newcommand\pln{\tilde{\phi}^{[\bc_0]}}
\newcommand\plnoned{\tilde{\phi}^{[c_0]}}
\newcommand\spaceH{\bbH^{[\bc_0]}}
\newcommand\spaceHoned{\bbH^{[c_0]}}
\newcommand\spaceHn{\tilde{\bbH}^{[\bc_0]}}
\newcommand\spaceHnoned{\tilde{\bbH}^{[c_0]}}
\newcommand\pd[2]{\dfrac{\partial {#1}}{\partial {#2}}}

\newcommand\inner[2]{\left\langle{#1},{#2}\right\rangle_{\spaceH_N}}
\newcommand\innern[2]{\left\langle{#1},{#2}\right\rangle_{\spaceHn_N}}
\newcommand\moment[1]{{\langle#1\rangle}}

\numberwithin{equation}{section}

\newcommand\delete[1]{}

\title{A Nonlinear Moment Model for Radiative Transfer Equation}

\author{ Ruo Li\thanks{CAPT, LMAM \& School of Mathematical Sciences,
    Peking University, Beijing, China, email: {\tt
      rli@math.pku.edu.cn}},~~
      Peng Song\thanks{Institute of Applied Physics and Computational
      Mathematics, Beijing, China, email: {\tt
      song\_peng@iapcm.ac.cn}} 
      ~and~ Lingchao Zheng\thanks{School of
    Mathematical Sciences, Peking University, Beijing, China, email:
    {\tt lczheng@pku.edu.cn}} }

\date{\today}

\begin{document}
\maketitle

\begin{abstract}
  We derive a nonlinear moment model for radiative transfer equation
  in 3D space, using the method to derive the nonlinear moment model
  for the radiative transfer equation in slab geometry. The resulted 3D
  \HMPN model enjoys a list of mathematical advantages, including
  global hyperbolicity, rotational invariance, physical wave speeds,
  spectral accuracy, and correct higher-order Eddington
  approximation. Simulation examples are presented to validate the new
  model numerically.

  \vspace*{4mm}
  \noindent {\bf Keywords:} Radiative transfer equation; moment
  method; nonlinear model; global hyperbolicity.

\end{abstract}


\section{Introduction}
The radiative transfer equation (RTE) depicts the motion of photons
and their interaction with the background medium. It has lots of
applications, such as radiation astronomy \cite{mihalas1978stellar},
optical imaging \cite{klose2002optical, tarvainen2005hybrid}, neutron
transport in reactor physics \cite{pomraning1973equations,
  duderstadt1979transport}, light transport in atmospheric radiative
transfer \cite{marshak20053d} and heat transfer
\cite{koch2004evaluation}. Due to the integro-differential form and
the high-dimensionality of RTE, how to develop efficient methods to
solve it numerically is an important but challenging topic. So far,
the commonly used numerical methods can be categorized into two groups: the
probabilistic methods, like the direct simulation Monte Carlo (DSMC) method
\cite{fleck1971, Bird, hayakawa2007coupled, densmore2012hybrid,
  abdikamalov2012new}, and the deterministic methods
\cite{broadwell1964study, larsen2010advances,
  Stamnes1988Electromagnetic, jeans1917stars, Davison1960on,
  dubroca1999theoretical, minerbo1978maximum,
  alldredge2016approximating, fan2018fast, MPN, HMPN}, such as the
discrete ordinates method ($S_N$) \cite{broadwell1964study,
  larsen2010advances, Stamnes1988Electromagnetic}, the moment methods
\cite{jeans1917stars, Davison1960on, dubroca1999theoretical,
  alldredge2016approximating, MPN, HMPN} and etc.

The discrete ordinates method ($S_N$) is one of the most popular
numerical methods to simulate the RTE, which solve the RTE along with
a discrete set of angular directions from a given quadrature
set. However, the \SN model assumes that the particles can only move
along the directions in the quadrature set, thus once the coordinate
system is rotated, the results of the \SN model can be different. The
lack of rotational invariance results in numerical artifacts, known as
\emph{ray effects} \cite{larsen2010advances}. 

In order to reduce the complexity of the RTE, the moment method
focuses on the evolution of a finite number of moments of the specific
intensity, which avoids the high-dimensionality of directly solving the
RTE. Since the governing equation of a lower order moment commonly
contains higher order moments, the moment system is often not
automatically closed. Hence one has to take a \emph{moment closure} to
close the moment system. A practical method for the moment closure is
to construct an ansatz to approximate the specific intensity. The
pioneer works in moment method include the spherical harmonics method
($P_N$) \cite{pomraning1973equations} and the maximum entropy method
($M_N$) \cite{levermore1996moment, dubroca1999theoretical,
  minerbo1978maximum}. The \PN model constructs the ansatz using
spherical harmonic polynomials. It can be regarded as a polynomial
expansion of the specific intensity around the equilibrium, which is
a constant function. One of the flaws is that the resulting system may
lead to nonphysical oscillations, or even worse, negative particle
concentration \cite{brunner2001one, brunner2005two,
  mcclarren2008solutions}. The \MN model constructs the ansatz using
the principle of maximum entropy, as the maximum entropy closure for
Boltzmann equation \cite{levermore1996moment,
  dubroca1999theoretical}. Unfortunately, no explicit expression of
the moment closure for \MN model can be given when the order
$N \geq 2$. To implement the model numerically, one has to solve an
ill-conditioned optimization problem to obtain an approximate moment
closure. This almost prohibits the application of the \MN model.

Recently, a nonlinear moment model (called the \MPN model) was
proposed in \cite{MPN} for the RTE in slab geometry. This model takes
the ansatz of the \Mone model (the first order \MN model) as the
weight function, then constructs the ansatz by expanding the specific
intensity around the weight function in terms of orthogonal
polynomials in the velocity variables. Numerical examples in
\cite{MPN} demonstrated a quite promising performance as an improved
approximation of the intensity in comparison of the \PN model. 
The \MPN model was further improved in \cite{HMPN} by a
globally hyperbolic regularization following the framework developed
in \cite{Fan, Fan_new, framework, Fan2015}. We note that the
regularization in \cite{HMPN} is a subtle modification of the work in
\cite{framework} instead of a direct application. Otherwise, the
resulting system may change the \Mone model, which leads to a wrong
higher-order Eddington approximation. Eventually, the \HMPN model was
proposed in \cite{HMPN} with not only global hyperbolicity, but also a
physical higher-order Eddington approximation.

Encouraged by the elegant mathematical structure and the promising
numerical performance of the \HMPN model for RTE in slab geometry, we
in this paper try to extend the method to derive the \HMPN model for 3D
problems. The steps of the extension are clear while there are still
numerous difficulties. Fortunately, the 3D \Mone model is explicit,
which allow us to construct the ansatz to approximate the specific
intensity using again the weighted polynomials with the weight
function is the ansatz of the \Mone model. To construct the function space of
the weighted polynomials, we need to give the orthogonal polynomial basis
with respect to the weight function. For the slab geometry \cite{MPN},
this can be implemented by a simple Gram-Schmidt orthogonalization.
For the 3D case, we have to use quasi-orthogonal polynomials rather
than orthogonal polynomials. Otherwise, it can be extremely
involving to accomplish the calculation, which makes further analysis
to the resulted model prohibited. We propose a procedure to make a
quasi Gram-Schmidt orthogonalization to obtain the quasi-orthogonal
polynomials in explicit expressions. This provides us a 3D \MPN model
in explicit formation, which can be mathematically analyzed. To achieve
global hyperbolicity, we still adopt the method in \cite{HMPN} to
regularize the 3D \MPN model. Quite smoothly a globally hyperbolic 3D
\HMPN model is eventually attained with a list of fabulous
mathematical natures inherited from its 1D counterpart for the slab
geometry. The resulted model is rotational invariant, with wave speeds
not greater than that of light, spectral approximation accuracy, and
correct higher-order Eddington approximation. We carry out preliminary
numerical simulating using an abruptly splitting scheme to validate the
new 3D \HMPN model. Some numerical examples on typical problems are
presented with satisfactory performance.

The rest of this paper is arranged as follows. In
\Cref{sec:preliminiaries}, we briefly introduce the moment methods for
RTE, and review how the \MPN and the \HMPN model in slab geometry were
derived in \cite{MPN, HMPN}. In \cref{sec:model}, we derive the 3D
\MPN model and prove that the model is rotational invariant. The
hyperbolic regularization is applied to give the \HMPN model in
\cref{sec:hyper}. The model is analyzed in detail therein. In
\cref{sec:num}, we introduce the numerical scheme to carry out
numerical simulations and present some numerical examples. The paper
is then ended with a short conclusion remarks.


\section{Preliminary} \label{sec:preliminiaries}

To model radiative transfer, the governing equation is a
time-dependent equation of the \emph{specific intensity} $I$ as
\begin{equation} \label{eq:radiativetransfer}
  \frac{1}{c}\pd{I}{t}+\bOmega\cdot \nabla_{\bx} I = \mS(I),
\end{equation}
where $c$ is the speed of light, and the specific intensity
$I=I(t, \bx; \bOmega, \nu)$ depends on time $t\in\bbR^+$, the spatial
coordinate of the photon $\bx\in\bbR^3$, the velocity direction
$\bOmega \in \bbS^2$ and the frequency $\nu \in \mathbb{R}^+$. In this
paper, our study omits the independent variable $\nu$ that $I$ is a
function of $t$, $\bx$ and $\bOmega$ only. The right hand side
$\mS(I)$ denotes the actions by the background medium on the
photons. A form of $\mS(I)$ adopted commonly was given in \cite{Bru02,
  McClarren2008Semi} as
\begin{equation} \label{eq:source_pre}
  \mS(I) = -\sigma_t I +
  \frac{1}{4\pi}ac\sigma_a T^4 + \frac{1}{4\pi}\sigma_s \int_{\bbS^2} I
  \dd\bOmega + \frac{s}{4\pi},
\end{equation}
where $a$ is the radiation constant, and $s=s(t, \bx)$ is an isotropic
external source of radiation.  The scattering coefficient $\sigma_s$,
the absorption coefficient $\sigma_a$, and the material temperature $T(t, \bx)$
depend on time $t$ and the spatial position $\bx$.
The total opacity coefficient is $\sigma_t=\sigma_a+\sigma_s$. 

In case that the problems slab geometry and spherical symmetric
geometry are considered, the 3D RTE \eqref{eq:radiativetransfer} can
be simplified to 1D problem. Precisely, in the slab geometry, the
specific intensity depends only upon the single spatial coordinate $z$
and the single angular coordinate $\arccos\mu$, the angle between
$\bOmega$ and the $z$-axis. Then the specific intensity becomes
$I = I(z, \mu)$, and \eqref{eq:radiativetransfer} is simplified as
\begin{equation}\label{eq:slab-geometry}
  \dfrac{1}{c}\pd{I}{t} + \mu \pd{I}{z} = \mS(I).
\end{equation}
The spherical geometry with perfect symmetry is a slightly more
complicated case. The RTE, where the specific intensity depends upon
only on the distance from the origin $r=\Vert \bx\Vert$, and the
angular variable $\arccos\mu$, which is the angle between $\bOmega$
and $\bx$. In this case, $I=I(r,\mu)$, and the 3D RTE
\eqref{eq:radiativetransfer} is simplified as
\begin{equation}\label{eq:spherical-geometry}
  \dfrac{1}{c}\pd{I}{t} + \mu \pd{I}{r} +\dfrac{1-\mu^2}{r} \pd{I}{\mu} =
  \mS(I).
\end{equation}
In \cite{alldredge2016approximating, MPN, HMPN, LI2020285}, 
for slab geometry and spherical symmetric geometry, 
some nonlinear moment models had been derived with
global hyperbolicity and promising performance in handling problems
with fair extreme specific intensity functions. The major aim of this
paper is to develop models for 3D problems with similar techniques.

At first, we define the moments of the 3D specific intensity. Let
$\alpha\in\bbN^3$ be a 3D multi-index, i.e.
$\alpha = (\alpha_1, \alpha_2, \alpha_3)^T$, $\alpha_1,\alpha_2,\alpha_3\in\bbN$. 
We define a function of
$t$, and $\bx$, denoted by $\moment{I}_{\alpha}(t, \bx)$, as
\begin{equation} \label{eq:momentsdefine} \moment{I}_{\alpha}(t, \bx)
  \triangleq \int_{\bbS^2} \bm\Omega^{\alpha} I(t, \bx;
  \bm\Omega)\dd\bm\Omega, \quad \alpha \in \bbN^3,
\end{equation} 
where
$\bOmega^{\alpha} = \Omega_1^{\alpha_1} \Omega_2^{\alpha_2}
\Omega_3^{\alpha_3}$. We call that $\moment{I}_{\alpha}$ is {\it the
  $\alpha$-th moment} of the specific intensity $I$.

Notice that $\bOmega \in \bbS^2$ implies 
$\Vert \bOmega \Vert = 1$, thus one has
\[
  \sum_{d=1}^{3} \moment{I}_{\alpha+2e_d} = \moment{I}_{\alpha},
  \qquad \forall \alpha\in\bbN^3,
\]
where $e_d$ represents the multi-index whose $d$-th index is 1, and
the else two indexes are 0. Therefore, we only need to consider these
moments $\moment{I}_{\alpha}$, $\alpha \in \mI$, where $\mI$ is a set
of multi-indexes, defined by
\begin{equation}
  \label{eq:setmultiindex}
  \mI \triangleq \left\{ \alpha: \alpha\in\bbN^3, \alpha_3\leq 1 \right\}. 
\end{equation}
The {\it order of the multi-index} $\alpha$ is defined as
$|\alpha|=\sum_{d=1}^{3} \alpha_d$, and we denote that
$\mI_N \triangleq \left\{ \alpha: \alpha \in \mI, |\alpha| \leq N
\right\}$. It is clear that once
$\left\{ \moment{I}_{\alpha} : \alpha \in \mI_N \right\}$ is
determined, one can obtain
$\left\{\moment{I}_{\alpha} : \alpha \in \bbN^3, |\alpha| \leq N
\right\}$. This allows us to discuss $\moment{I}_{\alpha}$ for
$\alpha\in\mI$ only to derive reduced models.

Multiplying \eqref{eq:radiativetransfer} by $\bOmega^{\alpha}$, and
taking the integration with respect to $\bOmega$ over $\bbS^2$, one can
have
\begin{equation} \label{eq:momentequations}
  \frac{1}{c}\pd{\moment{I}_{\alpha}}{t} + \sum_{d=1}^{3}\pd{
    \moment{I}_{\alpha+e_d}}{x_d} = \moment{\mS(I)}_{\alpha}, \quad
  {\alpha} \in \mI.
\end{equation}

In order to derive a moment model for \eqref{eq:radiativetransfer}, we
first truncate the system by discarding all the governing equations of
high order moments $\moment{I}_{\alpha}$, where $|\alpha|>N$, for a
given integer $N \in \bbN$. The truncated moment system is
\begin{equation} \label{eq:truncatedmomentequations} 
  \frac{1}{c}\pd{\moment{I}_{\alpha}}{t}
  + \sum_{d=1}^{3}\pd{ \moment{I}_{\alpha+e_d}}{x_d} = \moment{\mS(I)}_{\alpha},
  \quad {\alpha} \in \mI_N.
\end{equation}
However, the governing equations of $\moment{I}_{\alpha}$,
$| \alpha | = N$, involve three $N+1$ order moments
$\moment{I}_{\alpha+e_d}$ for $d=1,2,3$, thus the truncated system
\eqref{eq:truncatedmomentequations} is not closed.  Therefore, we need to
determine all these moments,
$\left\{ \moment{I}_{\alpha}, \alpha \in \mI_{N+1} \right\}$ to make
this truncated system \eqref{eq:truncatedmomentequations} closed.

We divide the moments in
$\left\{ \moment{I}_{\alpha}, \alpha \in \mI_{N+1} \right\}$ into two
parts: lower-order moments (also referred as {\it known moments} later
on), and higher-order moments (also referred as {\it unknown moments}
later on).
\begin{equation}
  \begin{array}[H]{ccc}
    \moment{I}_{\alpha}, \alpha\in\mI_{N+1}, |\alpha|\leq N
    & \qquad \qquad &\moment{I}_{\alpha}, \alpha\in\mI_{N+1}, |\alpha|=N+1\\
    \Downarrow & & \Downarrow \\
    \text{lower-order moments}&&\text{higher-order moments}\\
    \text{({\it known moments})}&& \text{({\it unknown moments})}
  \end{array}
\end{equation}
The aim of the so-called \emph{moment closure} to this system is to
approximate the unknown moments as functions of known moments, saying
to give a formulation as
\begin{equation}\label{eq:momentclosure}
  \moment{I}_{\alpha} \approx E_{\alpha} = E_{\alpha}(\moment{I}_{\beta},
  \beta\in \mI_N ), \quad \text{for }
  \alpha\in\mI_{N+1}, |\alpha|=N+1.
\end{equation}

To achieve this goal, a practical approach is to construct an ansatz
for the specific intensity. Precisely, let $E_{\alpha}$,
$\alpha\in\mI_N$, be the known moments for a certain unknown specific
intensity $I$. Then one may propose an expression
$\hat{I}(\bOmega;E_{\alpha}, \alpha\in\mI_N)$, called an \emph{ansatz}
to approximate $I$, such that
\begin{equation} \label{eq:momentrelation}
  \moment{\hat{I}(\cdot;E_{\alpha},\alpha\in\mI_{N})}_{\alpha} = E_{\alpha},
  \quad \alpha\in\mI_{N}.
\end{equation}
Often we require that $\hat{I}$ is uniquely determined by the
consistency relations \eqref{eq:momentrelation}. With $\hat{I}$ given,
the higher-order moments of $I$ are then approximated by the
higher-order moments of $\hat{I}$, i.e.,
\begin{equation}\label{eq:closure}
  E_{\alpha} = \moment{\hat{I}(\cdot; E_{\beta}, \beta\in\mI_{N})}_{\alpha}, \quad 
  \alpha\in\mI_{N+1}, |\alpha|=N+1.
\end{equation}
Therefore, one may take the closed moment system
\begin{equation}
  \label{eq:momentsystem}
\frac{1}{c} \pd{E_\alpha}{t} +\sum_{d=1}^{3} \pd{E_{\alpha+e_d}}{x_d} = \langle \mS(\hat{I})
\rangle_{\alpha}, \quad \alpha\in\mI_{N},
\end{equation}
as the reduced model to approximate the original RTE, where
$E_{\alpha+e_d}$ are functions of $E_{\alpha}$, $\alpha\in\mI_{N}$,
defined in \eqref{eq:momentclosure} and \eqref{eq:closure}.

Many existing models can be regarded as consequences using this moment
closure approach. For example, the \PN model \cite{jeans1917stars},
the \MN model \cite{levermore1996moment, dubroca1999theoretical}, the
positive \PN model \cite{hauck2010positive}, the $B_2$ model
\cite{alldredge2016approximating}, and the \MPN model \cite{MPN} are
in this fold. The \MPN model we proposed in \cite{MPN}, and then
improved in \cite{HMPN} as the \HMPN model, is limited for problems in
slab geometry, where it exhibits satisfactory numerical performance
for some standard benchmarks. To extend the method therein to 3D RTE,
below we first briefly review the methods in \cite{MPN, HMPN} to
derive models in slab geometry to clarify our idea.

The \MPN model derived in \cite{MPN} for RTE in slab geometry is based
on a method to combine the \PN model and the \MN model, which was
implemented by expanding the specific intensity around the ansatz of
the \Mone model in terms of orthogonal polynomials. The ansatz of the
\Mone model in slab geometry is
\begin{equation}\label{eq:mone_1dmpn}
  \hat{I}_{M_1}= \frac{\varepsilon}{(1+c_0\mu)^4}, 
\end{equation}
where $\varepsilon$ and $c_0$ are determined by the $0$-th moment
$E_0$ and $1$-th moment $E_1$, formulated as
\begin{equation} \label{eq:alphaandsigma} c_0 =
  -\frac{3E_1/E_0}{2+\sqrt{4-3(E_1/E_0)^2}},\quad \varepsilon =
  \frac{3(1-c_0^2)^3}{2(3+c_0^2)}.
\end{equation}
Then the specific intensity is approximated by a weighted
polynomial, with the weight function
\begin{equation}\label{eq:weight_1dmpn}
  \weightoned(\mu)= \frac{1}{(1+c_0\mu)^4}.
\end{equation}
The function space of the weighted polynomials is
\begin{equation}
  \label{eq:spaceH_1dmpn}
  \spaceHoned_{N} \triangleq \left\{ \weightoned\sum_{k=0}^{N}g_k\mu^k \right\}.
\end{equation}
Then the ansatz of the \MPN model is written as  
\begin{equation}\label{eq:ansatz_1dmpn}
  \hat{I}(\mu; E_0, \cdots, E_N) \triangleq \sum_{i=0}^N f_i \Ploned_i(\mu)
  \in\spaceHoned_{N},
  \quad \Ploned_i(\mu) = \ploned_i(\mu)\weightoned(\mu), 
\end{equation}
where $\Ploned_i(\mu)=\ploned_i(\mu)\weightoned(\mu)$,
$i = 0,1,\cdots,N$, are the basis functions, $\ploned_i(\mu)$ are
orthogonal polynomials with respect to the weight function, and $f_i$
are the expansion coefficients.

The orthogonal polynomials $\ploned_k(\mu)$ can be calculated by a
simple Gram-Schmidt orthogonalization, formulated as
\begin{equation}\label{eq:Gram-Schmidt_1dmpn}
  \ploned_0(\mu) = 1, \quad \ploned_j(\mu) = \mu^j - \sum_{k=0}^{j-1}
  \frac{\mK_{j,k}} {\mK_{k,k}} \ploned_k(\mu), \quad j\geq 1,
\end{equation}
where
$\mK_{j,k} = \int_{-1}^{1} \mu^j\ploned_k(\mu)\weightoned(\mu)\dd\mu$,
calculated by
\begin{equation}\label{eq:Gram-Schmidt_mk_1dmpn}
  \mK_{0,0}=\moment{\weightoned(\mu)}_{0},\quad
  \mK_{i,j}=\moment{\weightoned(\mu)}_{i+j} - \sum_{k=0}^{j-1}
  \frac{\mK_{j,k} \mK_{i,k}}{\mK_{k,k}}, \quad 1\leq j\leq i.
\end{equation}

Furthermore, by 
\begin{equation}\label{eq:fk_Ek_1dmpn}
  f_0 = \dfrac{E_0}{\mK_{0,0}},\quad
  f_i = \frac{1}{\mK_{i,i}} \left(E_i - \sum_{j=0}^{i-1}
    \mK_{i,j} f_j\right),\quad 1\leq i\leq N,
\end{equation}
one can determine the coefficients $f_k$, and the ansatz
$\hat{I}(\mu;E_0,E_1,\cdots E_N)$. Finally, we have the moment closure
as
\begin{equation}\label{eq:closureMPN_1dmpn}
E_{N+1} = \sum_{k=0}^N f_k \mK_{N+1,k}.
\end{equation}
Meanwhile, in the viewpoint of orthogonal projection, if we define the
orthogonal projection to function space $\spaceHoned_{N}$,
\begin{equation}
  \label{eq:projection_1dmpn}
  \mP_N: f = \sum_{k=0}^{+\infty}f_k\Ploned_k(\mu) \rightarrow 
  \mP_N f = \sum_{k=0}^{N}f_k\Ploned_k(\mu),
\end{equation}
then the \MPN moment system can be written as  
\begin{equation}
  \label{eq:mpn_system_projection_1dmpn}
  \dfrac{1}{c}\mP_N \pd{\mP_N I}{t} + \mP_N \mu \pd{\mP_N I}{z} =
  \mP_N \mS(\mP_N I).
\end{equation}

The weight function \eqref{eq:weight_1dmpn} permits the \MPN model to
approximate a strongly anisotropic distribution with very
high accuracy. In \cite{HMPN}, the \MPN model was further improved by
a hyperbolic regularization, which provides global hyperbolicity. 
To achieve the required hyperbolicity, the model
reduction framework in \cite{framework, Fan2015} suggests adding one
more projection between the operators $\mu\cdot$ and $\pd{\cdot}{z}$
in \eqref{eq:mpn_system_projection_1dmpn} to regularize the \MPN model
to be globally hyperbolic, and the resulting model is
\begin{equation}\label{eq:ms_framework_1dmpn}
  \frac{1}{c}\mP_N\pd{\mP_N I}{t} + \mP_N\mu\mP_N\pd{\mP_N I}{z} =
  \mP_N\mS(\mP_N I).
\end{equation} 
An interesting point observed in \cite{HMPN} is that this
regularization changes the \MPN model when $N=1$. It is definitely
inappropriate that the \Mone model is changed by the
regularization. In order to fix this defect, another weight function 
is introduced as 
\begin{equation}\label{eq:weightn_1dmpn}
  \weightnoned=\frac{1}{(1+c_0\mu)^5},
\end{equation}
and then a new function space is defined as
\begin{equation}\label{eq:spaceHn_1dmpn}
  \spaceHnoned_N \triangleq \left\{ \weightnoned \sum_{k=0}^{N}
    g_k\mu^k\right\}.
\end{equation}
In this subspace, one can define
the orthogonal polynomials $\plnoned_k(\mu)$ and the basis function
$\Plnoned_k(\mu)$, $0\leq k\leq N$. Furthermore, a new projection $\mPn_N$ is defined as
\begin{equation}
  \label{eq:projection_1dhmpn}
  \mPn_N: f = \sum_{k=0}^{+\infty}f_k\Plnoned_k(\mu) \rightarrow 
  \mPn_N f = \sum_{k=0}^{N}f_k\Plnoned_k(\mu).
\end{equation}
The new hyperbolic regularization in \cite{HMPN} adds one more
projection $\mPn_N$ to give the \HMPN model formulated as
\begin{equation}\label{eq:hmpn_system_projection_1dmpn}
  \frac{1}{c}\mPn_N\pd{\mP_N I}{t} + \mPn_N\mu\mPn_N\pd{\mP_N I}{z}
  = \mPn_N\mS(\mP_N I).
\end{equation} 
It was revealed that the \HMPN model enjoys some desired properties,
such as
\begin{property}\label{pro:hmpn_1dmpn}
  \begin{enumerate}
  \item The \HMPN model is globally hyperbolic.
    \item The characteristic speeds of the \HMPN model lie in
      $[-c, c]$.
    \item The regularization vanishes for the case $N=1$.
    \item Between the \MPN model and the \HMPN model, the governing
      equation of $E_k$, $k=0,\dots,N-1$ is not changed.
  \end{enumerate}
\end{property}


\section{Moment Model Reduction}\label{sec:model}

In this section, we adopt the strategy introduced in
\Cref{sec:preliminiaries} to derive a \MPN type model for the 3D RTE
at first.

\subsection{Formal derivation}
\label{sec:Gram-Schmidt_orthogonalization}
Let us start with the \Mone model in 3D case. The ansatz of specific
intensity is
\begin{equation}
  \label{eq:ansatzofM1}
  \hat{I}_{M_1}(\bm\Omega; E_0, E_{e_1}, E_{e_2}, E_{e_3}) =
  \dfrac{\varepsilon}{(1+\bc_0\cdot\bOmega)^4}.
\end{equation}
with the known moments $E_0$, $E_{e_1}$, $E_{e_2}$, and $E_{e_3}$. We
denote them as $\bm E_0=(E_0)^T$, and
$\bm E_1=(E_{e_1},E_{e_2},E_{e_3})^T$, respectively. Direct
calculation yields that 
\begin{equation}
  \label{eq:bc0_expression}
  \bc_0 = \dfrac{-2+\sqrt{4-3(\Vert \bm
      E_1\Vert/E_0)^2}}{\Vert \bm E_1\Vert / E_0} \dfrac{\bm E_1}{\Vert
    \bm E_1\Vert},
\end{equation}

Following \cite{MPN}, we approximate the specific intensity with a
weighted polynomial, and the weight function is chosen as the ansatz
of the \Mone model, as \eqref{eq:ansatzofM1}. For simplicity, we take
the weight function as
\[
  \weight(\bm\Omega) = \dfrac{1}{(1+\bc_0\cdot\bOmega)^4}.
\]
Then the function space of weighted polynomials is defined as
\begin{equation}
  \label{eq:spaceH}
  \spaceH_N \triangleq \left\{ \weight(\bOmega) \sum_{\alpha\in\mI_{N}} 
    g_{\alpha}\bm\Omega^{\alpha} \right\}.
\end{equation}
The ansatz of the 3D \MPN model $\hat{I}(\bOmega)$ is chosen to
satisfy
\[
  \hat{I}\in\spaceH_N,\quad \langle \hat{I} \rangle_{\alpha} =
  E_{\alpha}, \quad \alpha \in \mI_{N}.
\] 

In order to determine the ansatz $\hat{I}$, we first rewrite $\hat{I}$
as
\begin{equation}
  \label{eq:ansatz}
  \hat{I}(t,\bx;\bOmega) 
  = \hat{I}(t; \bm{E})
  = \sum_{\alpha\in\mI_{N}} f_{\alpha}(t,\bx)
  \pl_{\alpha}(\bOmega)\omega^{[\bc_0(t,\bx)]}(\bOmega) \in \spaceH_N,
\end{equation}
where $\pl_{\alpha}, \alpha\in\mI_{N}$ are quasi-orthogonal
polynomials with respect to the weight function $\weight$, and
$f_{\alpha}$ are the corresponding coefficients. On the
quasi-orthogonal polynomials, we define the inner product
\begin{equation}
  \label{innerH}
  \inner{f}{g} \triangleq \int_{\bbS^2} fg\weight\dd\bOmega,
\end{equation}
then the quasi-orthogonal polynomials $\pl_{\alpha}$ satisfy that
\begin{equation}
  \label{eq:orthogonalpolynomial}
\inner{\pl_{\alpha}}{\pl_\beta} = \int_{\bbS^2} \pl_{\alpha}(\bOmega)\pl_{\beta}(\bOmega)
  \weight(\bOmega) \dd\bOmega = 0, \quad
  \text{when } |\alpha|\neq |\beta|.  
\end{equation}
Moreover, the quasi-polynomial $\pl_{\alpha}$ satisfies that its only
$|\alpha|$-order term is $\bOmega^{\alpha}$, i.e.  
\begin{equation}
  \label{eq:monicpolynomial}
\pl_{\alpha} = \bOmega^{\alpha} + \sum_{\beta\in\mI_{|\alpha|-1}}
h_{\alpha,\beta}\bOmega^{\beta}.
\end{equation}
For later usage, we denote the quasi-orthogonal functions 
\begin{equation}
  \label{eq:orthogonalfunctions}
  \Pl_{\alpha} = \pl_{\alpha}\weight.
\end{equation}
\begin{remark}
  Let us remark that $\pl_{\alpha}$, $\alpha\in\mI$ are quasi-orthogonal
  polynomials, rather than orthogonal polynomials. In other words,
  $\inner{\pl_{\alpha}}{\pl_{\beta}} $ can be non-zero for
  $|\alpha|=|\beta|$ and $\alpha\neq \beta$.
\end{remark}

Let us construct these quasi-orthogonal polynomials $\pl_{\alpha}$ by
the Gram-Schmidt orthogonalization. This Gram-Schmidt
orthogonalization is different from the 1D case used in \cite{MPN}. We
denote $\mE_{\alpha}$ as the moments of the weight function
$\mE_{\alpha} \triangleq \langle \weight\rangle_{\alpha}$ for later
usage.

In \eqref{eq:bc0_expression}, we denote $\bm E_{0}$ and $\bm E_{1}$ as
the vector of moments of order $0$ and order $1$.  In the following
discussion, we will continue to use this notation, i.e. we denote
$\bm E_k$ as the vector of moments of order $k$, whose dimension is
$2k+1$.  Similarly, we can denote $\bm{\pl}_k$, $\bm{\Pl}_k$, $\bm f_k$ and
$\bm \mE_k$. Using this notation, the moment closure
\[
  E_{\alpha}(E_\beta, \beta\in\mI_{N}), \quad |\alpha| = N+1,
  \alpha\in\mI
\]
can be rewritten as 
\[  
  \bm E_{N+1}(\bm E_k,\ 0\leq k\leq N).
\]
We denote by $\bm\varphi_k$ to be the vector of $\bOmega^{\alpha}$,
$\alpha\in\mI$ and $|\alpha|=k$. Using these notations, the
ansatz \eqref{eq:ansatz} is rewritten as 
\begin{equation}
  \label{eq:ansatz_vectorform}
  \hat{I} = \sum_{k=0}^{N} \bm f_k \bm\Pl_k = \weight\sum_{k=0}^{N} \bm f_k \bm\pl_k.
\end{equation}
We denote the inner product as
\[
  \mK_{\alpha, \beta} \triangleq
  \inner{\bOmega^{\alpha}}{\pl_{\beta}}, 
\]
according to the properties of the quasi-orthogonal polynomials
$\pl_{\alpha}$ \eqref{eq:orthogonalpolynomial} and
\eqref{eq:monicpolynomial}, we have
\[ 
  \mK_{\alpha,\beta} = 0, \text{ when } |\alpha|<|\beta|; \quad
  \mK_{\alpha,\beta} = \inner{\pl_{\alpha}}{\pl_{\beta}}
  =\mK_{\beta,\alpha}, \text{ when } |\alpha|=|\beta|.
\]
Let us denote the matrix composed of
\[ 
  \mK_{\alpha,\beta}, \quad \alpha,\beta\in\mI_{N}, |\alpha|=i,
  |\beta|=j,
\]
as $\bm\mK_{i,j}$, whose dimension is $(2i+1)\times(2j+1)$. It is not
difficult to show that
\begin{enumerate}
  \item $\bm\mK_{i,j} = 0$, when $i<j$.
  \item $\bm\mK_{k,k}$ is symmetric and positive definite. 
\end{enumerate}
Then, by the Gram-Schmidt orthogonalization, the calculation of the
quasi-orthogonal polynomials can be written as
\begin{equation}
  \label{eq:Gram-Schmidt}
  \bm\pl_{j} = \bm\varphi_j - \sum_{k=0}^{j-1} 
  \bm\mK_{j,k}\bm\mK_{k,k}^{-1} \bm\pl_{k}, \quad j\geq 0.
\end{equation}
Taking the inner product of $\bm\varphi_i$ and the transpose of
\eqref{eq:Gram-Schmidt}, one can derive the recursion relation of
coefficients $\mK_{\alpha, \beta}$,
\begin{equation}
  \label{eq:mK_Gram-Schmidt}
  \bm\mK_{i,j} = \inner{\bm\varphi_i}{\bm\varphi_j^T} - 
  \sum_{k=0}^{j-1} \bm\mK_{i,k}\bm\mK_{k,k}^{-1}\bm\mK_{j,k}^T.
\end{equation}
Therefore, once $\mE_{\alpha}$ are calculated, $\mK_{\alpha,\beta}$ is
determined by \eqref{eq:mK_Gram-Schmidt}, and $\pl_{\alpha}$ is
determined by \eqref{eq:Gram-Schmidt}. We note that all these
formulations are explicit. The details of the calculation of
$\mE_{\alpha}$ are presented in \Cref{sec:calculationofmoments}.

Since the quasi-orthogonal polynomials have been calculated, one can
simply determine the coefficients $f_{\alpha}$ by the constraints of
the known moments. To be precise,
$\langle \hat{I}\rangle_{k} = \bm E_{k}$ implies that
\begin{equation}
  \label{eq:contraint_Ek}
  \sum_{j=0}^{k} \bm\mK_{k,j}\bm f_j = \bm E_{k}, \quad 0\leq k \leq
  N,
\end{equation}
thus the coefficients $f_{\alpha}$, $\alpha\in\mI_{N}$ are obtained by
\begin{equation}
  \label{eq:coefficientscalcul}
  \bm f_k = \bm\mK_{k,k}^{-1}\left( \bm{E}_k 
  - \sum_{j=0}^{k-1} \bm\mK_{k,j}\bm f_j \right),\quad 0\leq k\leq N. 
\end{equation}
Eventually, the moment closure of the \MPN model is given as
\begin{equation}
  \label{eq:momentclosure_vector}
  \bm E_{N+1} = \moment{\hat{I}}_{N+1} = \sum_{j=0}^{N} \bm \mK_{N+1,j} \bm f_j.
\end{equation}
Substituting \eqref{eq:momentclosure_vector} into
\eqref{eq:momentsystem}, a closed moment system is attained. We 
refer this model as the {\it 3D \MPN moment system} later on.
\begin{remark}
  \label{rem:f1}
  Notice that $\bc_0$ and the weight function $\weight$ are determined
  by $\bm E_0$ and $\bm E_1$, therefore
  \[
    \dfrac{\moment{\hat{I}}_{e_d}}{\moment{\hat{I}}_0} =
    \dfrac{E_{e_d}}{E_0} =
    \dfrac{\moment{\weight}_{e_d}}{\moment{\weight}_0} =
    \dfrac{\mK_{e_d,0}}{\mK_{0,0}} ,\quad d=1,2,3.
  \]
  Simple calculation yields 
  \[
    f_0 = E_0/\mK_{0,0},\quad E_{e_d} - \mK_{e_d,0}f_0 = 0, \quad
    d=1,2,3,
  \]
  which tells that 
  \begin{equation}
    \label{eq:f1}
    f_{e_d} = 0, \quad  d= 1,2,3.
  \end{equation}
\end{remark}

It is essential for a reduced model to preserve the Galilean
invariance of 3D RTE \eqref{eq:radiativetransfer}. It is often trivial
to preserve the reduced model to be invariant under a translation. However, only
both with translational invariance and rotational invariance, the
reduced model in 3D is Galilean invariant. Unfortunately, the
rotational invariance is not usually preserved by the model reduction
of 3D RTE automatically. For instance, the extensively used $S_N$
model, due to the lack of the rotational invariance, leads to the
so-called {\it ray-effect} in numerical simulations, which is regarded
as its major flaw \cite{larsen2010advances}.

Thanks to the rotational invariance of the weight function (the 3D \Mone
model), the rotational invariance of the 3D \MPN model is trivial to be
verified. We denote the orthogonal coordinate system as
$(\be_x, \be_y, \be_z)$, and the coordinate of an element $\bm p$ is
$\bx = (\bm p\cdot \be_x, \bp\cdot\be_y, \bp\cdot\be_z)^T$. After a
given rotation, the orthogonal coordinate system becomes
$(\overline{\be_x}, \overline{\be_y}, \overline{\be_z})$, and there is
a constant $3\times 3$ orthogonal matrix $\bG_1$, with
$\text{Det}(\bG_1)=1$, satisfies that
$(\overline{\be_x}, \overline{\be_y}, \overline{\be_z})=(\be_x, \be_y,
\be_z)\bG_1^T$, then the coordinate of the same element $\bp$ is given
by
$\overline{\bx} = (\bp\cdot\overline{\be_x}, \bp\cdot\overline{\be_y},
\bp\cdot\overline{\be_z})^T = \bG_1 \bx$. Then the moments in the
rotated system can be written as a linear combination of the moments
in the original system. For any given $k\in\bbN$, there exists a
non-singular $(2k+1)\times(2k+1)$ matrix $\bG_k$, satisfying that 
\[
  \overline{\bm\varphi_k} = \bG_k \bm\varphi_k, \quad k\geq 0,
\]
where $\bm\varphi_k$ and $\overline{\bm\varphi_k}$ are the vectors of
the $k$-th moment of the weight function, defined in
\Cref{sec:Gram-Schmidt_orthogonalization}, 
in the original system and the rotated system,
respectively.
Furthermore, in the rotated system, the moments are denoted as
$\overline{\bm E_k}$, $0\leq k\leq N$. One can directly verify that
\[
  \overline{\bm E_k} = \bG_k \bm E_k, \quad 0\leq k\leq N.
\]
Let us give a lemma at first.
\begin{lemma}
  \label{lem:rotationalinvariance_lemma}
  \[
    \mathcal{D}_{\bG_1 \bx} (\bG_{k+1} \bm{E}_{k+1}) = \bG_k
    \mathcal{D}_{\bx} (\bm{E}_{k+1}), \quad k\geq 0,
  \]
  where $\mathcal{D}_{\bx}(\bm{E}_{k+1})$ is defined as a $(2k+1)$-vector 
  corresponding to $\bm{E}_k$. More accurately, if the
  $\mathcal{N}(\alpha,k)$-th element of $\bm E_k$ is $E_\alpha$,
  $|\alpha|=k$, then
  the $\mathcal{N}(\alpha,k)$-th element of
  $\mathcal{D}_{\bx}(\bm{E}_{k+1})$ is
  $\displaystyle \sum_{d=1}^{3} \pd{E_{\alpha+e_d}}{x_d}$.
\end{lemma}
\begin{proof}
  Noticing $\overline{\bm \varphi_k} = \bG_k{\bm \varphi_k}$, and the
  $\mathcal{N}(\alpha,k)$-th element of $\bm\varphi_k$ is
  $\bm\Omega^{\alpha}$, we have  
  \begin{equation}
    (\bG_1\bOmega)^{\alpha} = \sum_{|\beta|=k,
      \beta\in\mI}\bG_{k,\mathcal{N}(\alpha,k),
      \mathcal{N}(\beta,k)}\bOmega^{\beta},
    \quad |\alpha|=k, \alpha\in\mI,
    \label{eq:RI_omega}
  \end{equation}
  and $\overline{\bm E_k} = \bG_k \bm E_k$ can be rewritten as
  \begin{equation}
    \overline{E_{\alpha}} = \sum_{|\beta|=k,
      \beta\in\mI}\bG_{k,\mathcal{N}(\alpha,k),
      \mathcal{N}(\beta,k)}E_{\beta},
    \quad |\alpha|=k, \alpha\in\mI.
    \label{eq:RI_E}
  \end{equation}
  Moreover, multiplying $(\bG_1\bOmega)^{e_d}$ on \eqref{eq:RI_omega}, one
  has
  \[
    \begin{aligned}
      (\bG_1\bOmega)^{\alpha+e_d} = & \sum_{|\beta|=k,
        \beta\in\mI}\bG_{k,\mathcal{N}(\alpha,k),
        \mathcal{N}(\beta,k)}\bOmega^{\beta} (\bG_1 \bOmega)^{e_d} \\
      = &\sum_{l=1}^{3}\sum_{|\beta|=k, \beta\in\mI}\bG_{k,
      \mathcal{N}(\alpha,k),\mathcal{N}(\beta,k)}
      \bG_{1,d,l}\bOmega^{\beta+e_l}, \quad |\alpha|=k, \alpha\in\mI,
    \end{aligned}
  \]
  which is equivalent to
  \[
    \overline{E_{\alpha+e_d}} = \sum_{l=1}^{3}\sum_{|\beta|=k,
    \beta\in\mI}\bG_{k,\mathcal{N}(\alpha,k),\mathcal{N}(\beta,k)}
    \bG_{1,d,l}E_{\beta+e_l}, \quad |\alpha|=k, \alpha\in\mI, d=1,2,3.
  \]
  Noticing $\overline{x_d} = \sum_{s=1}^{3} \bG_{1,d,s}x_s$ and
  $\bG_1$ is a orthogonal matrix,  one has 
  $x_s = \sum_{d=1}^{3} \bG^T_{1,s,d} \overline{x_d}
  = \sum_{d=1}^{3} \bG_{1,d,s} \overline{x_d}$, and 
  \begin{equation}
    \begin{aligned}
      \sum_{d=1}^{3}\pd{\overline{E_{\alpha+e_d}}}{\overline{x_d}} &=
      \sum_{d=1}^{3}\sum_{s=1}^{3}
      \pd{\left(\sum_{l=1}^{3}\sum_{|\beta|=k,
      \beta\in\mI}\bG_{k,\mathcal{N}(\alpha,k),\mathcal{N}(\beta,k)}
          \bG_{1,d,l}E_{\beta+e_l}\right)}
          {x_s} \pd{x_s}{\overline{x_d}} \\
      &=
      \sum_{|\beta|=k,\beta\in\mI}\sum_{d,l,s=1}^{3}\bG_{k,\mathcal{N}
        (\alpha,k),\mathcal{N}(\beta,k)}\bG_{1,d,l}
      \pd{E_{\beta+e_l}}{x_s} \bG_{1,d,s} \\
      &= \sum_{l=1}^{3}\sum_{|\beta|=k,\beta\in\mI}
      \bG_{k,\mathcal{N}(\alpha,k),\mathcal{N}(\beta,k)}
      \pd{E_{\beta+e_l}}{x_l}.
    \end{aligned}
    \label{eq:RI_lemma_res}
  \end{equation}
  Compare \eqref{eq:RI_lemma_res} with \eqref{eq:RI_E}, one can yield
  \[
    \mathcal{D}_{\bG_1 \bx} (\bG_{k+1} \bm{E}_{k+1}) = \bG_k
    \mathcal{D}_{\bx} (\bm{E}_{k+1}), \quad k\geq 0.
  \]
\end{proof}
For $0\leq k\leq N-1$, the 3D \MPN system for the $k$-th order moment
can be written as 
\begin{equation}
  \dfrac{1}{c}\pd{\bm{E}_k}{t} + \mathcal{D}_{\bm x}(\bm E_{k+1}) =
  \bm S_k,
\end{equation}
and 
\begin{equation}
  \dfrac{1}{c}\pd{\overline{\bm{E}_k}}{t} 
  + \mathcal{D}_{\overline{\bm x}}(\overline{\bm{E}_{k+1}}) = \overline{\bm{S}_k}, 
  \label{eq:rotationalinvariance_eq}
\end{equation}
in the original coordinate system and rotated coordinate system,
respectively, where $\bm S_k$ and $\overline{\bm S_k}$ are the vectors
of $S_{\alpha}$, $|\alpha|=k$ in the original system and rotated system.
According to \Cref{lem:rotationalinvariance_lemma} and noticing
$\overline{\bm{E}_k} = \bG_k \bm E_{k}$, $\overline{\bm x}=\bG_1 \bm
x$, we have that \eqref{eq:rotationalinvariance_eq} is
equivalent to 
\[
  \dfrac{1}{c}\bG_k\pd{\bm E_k}{t}
  +\bG_k \mathcal{D}_{\overline{\bm x}}
  (\overline{\bm E_{k+1}}) = \bG_k \overline{\bm S_k},
\]
thus we obtain the rotational invariance for $0\leq k\leq N-1$.
Meanwhile, the governing equations of the $N$-th moment in the original
system and rotated system are
\begin{equation}
  \dfrac{1}{c}\pd{\bm{E}_N}{t} + \mathcal{D}_{\bm x}(\bm E_{N+1}) =
  \bm S_N,
\end{equation}
and 
\begin{equation}
  \dfrac{1}{c}\pd{\overline{\bm{E}_N}}{t} 
  + \mathcal{D}_{\overline{\bm x}}(\bm{E}_{N+1}(\overline{\bm{E}_k},
  0\leq k \leq N)) =
  \overline{\bm{S}_N}.
\end{equation}
Therefore in order to obtain the rotational invariance, 
one only needs to show that based on the
rotated known moments $\overline{\bm E_k}$, $0\leq k\leq N$, the
moment closure of the 3D \MPN model would give the rotated $(N+1)$-th
moments $\overline{\bm E_{N+1}}$, i.e.
\[
  \bm E_{N+1}(\overline{\bm E_{k}}, 0\leq k\leq N) = \overline{\bm
    E_{N+1}(\bm E_{k}, 0\leq k\leq N)} = \bG_{N+1} \bm E_{N+1}.
\]
Apart from the moments $\overline{\bm \varphi_k}$ and
$\overline{\bm E_k}$, in the following discussion, we denote the
variables with an overline $\overline{*}$ as the variables in the rotated
coordinate system, such as $\overline{\bc_0}$, $\overline{\weight}$,
$\overline{\bm f_k}$, and $\overline{\bm \mK_{i,j}}$. Then we arrive
the following theorem:
\begin{theorem} \label{thm:rotationalinvariance} The 3D \MPN model is
  rotational invariant.
\end{theorem}
\begin{proof}
  First, we consider the rotational invariance of the weight function.
  According to the expression of $\bc_0$ \eqref{eq:bc0_expression}, we
  have for the rotated system, $\overline{\bc_0} = \bG_1 \bc_0$, thus
  the weight function is
  \[
    \overline{\weight}(\bOmega)=\omega^{[\overline{\bc_0}]}(\bOmega) =
    \dfrac{1}{(1+\overline{\bc_0}\cdot\bOmega)^4}.
  \]
  Notice that
  $\overline{\bc_0}\cdot{\bG_1\bOmega}=\bc_0^T\bG_1^T\bG_1\bOmega =
  \bc_0\cdot\bOmega$, we have
  \begin{equation}
    \label{eq:weight_rotated}
    \overline{\weight}(\bG_1{\bOmega})=\weight(\bOmega).
  \end{equation}
  Moreover, on the moments of the weight function, we have 
  \begin{equation} 
    \label{eq:moments_rotated}
    \int_{\bbS^2} \bm\varphi_i \bm \varphi_j^T \overline{\weight(\bOmega)} \dd\bOmega 
    = \int_{\bbS^2} \bG_i \bm\varphi_i \bm\varphi_j^T \bG_j^T
    \weight(\bOmega) \dd\bOmega =
    \bG_i\inner{\bm\varphi_i}{\bm\varphi_j^T}\bG_j^T,
  \end{equation}
  where the first equality is according to \eqref{eq:weight_rotated}
  and using a variable substitution (replace $\bOmega$ by
  $\bG_1\bOmega$).  Therefore, according to \eqref{eq:mK_Gram-Schmidt}
  and \eqref{eq:moments_rotated}, using mathematical methods of
  induction, we have the coefficient matrix in the new coordinate
  system $\overline{\bm\mK_{i,j}}$ satisfies that
  \[
    \overline{\bm \mK_{i,j}} = \bG_i \bm\mK_{i,j} \bG_j^T.
  \]
  Analogously, by \eqref{eq:coefficientscalcul}, we have $\overline{\bm f_k} = \bG_k \bm f_k$. 
  Then by \eqref{eq:momentclosure_vector}, the moment closure is
  $\bG_{N+1}\bm E_{N+1} = \overline{\bm E_{N+1}}$, which ends the
  proof.
\end{proof}

\subsection{Implementation details}
\label{sec:calculationofmoments}

\paragraph{Calculation of $\mE_{\alpha}$}
In \Cref{sec:Gram-Schmidt_orthogonalization}, the calculation of the
Gram-Schmidt orthogonalization requires a practical method to obtain
$\mE_{\alpha}$. However, the calculation of $\mE_{\alpha}$ is not
trivial. In this subsection, we present the way to calculate
$\mE_{\alpha}$.

Let $\bc = \dfrac{\bc_0}{\Vert \bc_0\Vert}$, and $\ba, \bb, \bc$ is a
unit orthogonal basis of $\bbR^3$, then we can decompose $\bOmega$
under this basis, i.e.
\[
  \bOmega\cdot\ba = \sin\theta\cos\varphi,\quad \bOmega\cdot\bb =
  \sin\theta\sin\varphi,\quad \bOmega\cdot\bc = \cos\theta, \quad
  0\leq \theta<\pi,\ 0\leq\varphi<2\pi,
\]
then we have
$\bOmega = (\bOmega\cdot\ba)\ba + (\bOmega\cdot\bb)\bb +
(\bOmega\cdot\bc)\bc$, and
\[
  \mE_{\alpha} = \int_{\bbS^2} \bOmega^{\alpha}
  \weight(\bOmega)\dd\bOmega = 
  \int_{0}^{2\pi}\int_{0}^{\pi}
  \dfrac{\Omega_1^{\alpha_1}\Omega_2^{\alpha_2}\Omega_3^{\alpha_3}}
  {(1+\Vert \bc_0\Vert\cos\theta)^4}\sin\theta\dd\theta\dd\varphi.
\]
Notice that 
\[
  \begin{aligned}
    &\Omega_1 = a_1\sin\theta\cos\varphi + b_1\sin\theta\sin\varphi +
    c_1\cos\theta,\\
    &\Omega_2 = a_2\sin\theta\cos\varphi+ b_2\sin\theta\sin\varphi +
    c_2\cos\theta,\\
    &\Omega_3 = a_3\sin\theta\cos\varphi+ b_3\sin\theta\sin\varphi +
    c_3\cos\theta,\\
  \end{aligned}
\]
we have
\[
  \begin{aligned}
    \Omega_k^{\alpha_k} =&
    (a_k\sin\theta\cos\varphi+b_k\sin\theta\sin\varphi+c_k\cos\theta)^{\alpha_k}
    \\
    =& \sum_{n_k\leq m_k\leq \alpha_k}
    \dfrac{\alpha_k!}{(m_k-n_k)!n_k!(\alpha_k-m_k)!}
    (a_k\sin\theta\cos\varphi)^{m_k-n_k}
    (b_k\sin\theta\sin\varphi)^{n_k}
    (c_k\cos\theta)^{\alpha_k-m_k}  \\
    =& \sum_{n_k\leq m_k\leq \alpha_k}
    \dfrac{\alpha_k!}{(m_k-n_k)!n_k!(\alpha_k-m_k)!}  a_k^{m_k-n_k}
    b_k^{n_k} c_k^{\alpha_k-m_k}
    \sin^{m_k}\theta\cos^{\alpha_k-m_k}\theta
    \cos^{m_k-n_k}\varphi\sin^{n_k}\varphi.
  \end{aligned}
\]
Collecting the terms, we have
\[
  \begin{aligned}
    \bm\Omega^{\alpha} =& \prod_{k=1}^{3} \Omega_k^{\alpha_k}
    =\sum_{n\leq m\leq |\alpha|} C^{\alpha_1,\alpha_2,\alpha_3}_{m,n}
    \sin^{m}\theta
    \cos^{|\alpha|-m}\theta\sin^n\varphi\cos^{m-n}\varphi,
  \end{aligned}
\]
where $C^{\alpha_1, \alpha_2, \alpha_3}_{m,n}$ are coefficients.
Precisely, denote
$C^{\alpha}_{m,n} = C^{\alpha_1, \alpha_2, \alpha_3}_{m,n}$, then one
can obtain the recursion relationship of $C^{\alpha}_{m,n}$,
formulated as
\[
  C_{m,n}^{\alpha} = c_k C_{m,n}^{\alpha-e_k} +
  a_kC_{m-1,n}^{\alpha-e_k} +b_kC_{m-1,n-1}^{\alpha-e_k}, \quad
  k=1,2,3,
\]
Furthermore, we can obtain $\mE_{\alpha}$ by 
\[ 
  \mE_{\alpha} =\sum_{n\leq m\leq |\alpha|}
  {C^{\alpha_1,\alpha_2,\alpha_3}_{m,n}} \int_{0}^{\pi}
  \dfrac{\sin^{m+1}\theta\cos^{|\alpha|-m}\theta}{(1+c_0\cos\theta)^4}\dd\theta
  \int_{0}^{2\pi} \sin^{n}\varphi\cos^{m-n}\varphi\dd\varphi.
\]
Direct calculation yields
\[
  I_{|\alpha|,m}=\int_{0}^{\pi}
  \dfrac{\sin^{m+1}\theta\cos^{|\alpha|-m}\theta}{(1+c_0\cos\theta)^4}\dd\theta
  =\int_{-1}^{1} \dfrac{(1-\mu^2)^{m/2}\mu^{|\alpha|-m}}{(1+c_0\mu)^4}
  \dd\mu = \sum_{l=0}^{m/2}\binom{m/2}{l}(-1)^l M_{|\alpha|-m+2l},
\]
where $M_s = \int_{-1}^{1} \dfrac{\mu^s}{(1+c_0\mu)^4}\dd\mu$ has
already been calculated in the 1D \MPN model \cite{MPN}. On the other
hand,
\[
  J_{n,m-n}=\int_{0}^{2\pi} \sin^n\varphi\cos^{m-n}\varphi\dd\varphi=
  \left\{
    \begin{aligned}
      &\dfrac{2\Gamma\left( \frac{1+m-n}{2} \right)\Gamma\left(
          \frac{1+n}{2}
        \right)}{\Gamma\left( \frac{m+2}{2}\right)}, \quad
      & m,n \text{ are even},\\
      &0,\quad & \text{otherwise.}
    \end{aligned}
  \right. 
\]
Moreover, when $m$ and $n$ are even, 
\[ 
  J_{n,m-n} = \dfrac{2\Gamma\left( \frac{1+m-n}{2} \right)\Gamma\left(
      \frac{1+n}{2} \right)}{\Gamma\left( \frac{m+2}{2}\right)} = 2\pi
  \frac{(m-n-1)!!(n-1)!!}{2^{\frac{m-n}{2}}2^{\frac{n}{2}}(\frac{m}{2})!}
  = 2\pi
  \frac{(m-n-1)!!(n-1)!!}{2^{\frac{m}{2}}\left(\frac{m}{2}\right)!}.
\]
Therefore, we have
\[
  \mE_{\alpha} = \sum_{n\leq m\leq |\alpha|}
  C^{\alpha_1,\alpha_2,\alpha_3}_{m,n}I_{|\alpha|,m}J_{n,m-n}.
\]

\begin{remark}
  The choice of $\ba$, $\bb$, and $\bc$ can be casual.  For example,
  if
  $\bc = (\sin\theta_0 \cos\varphi_0, \sin\theta_0\sin\varphi_0,
  \cos\theta_0)^T$, then $\ba$, $\bb$, and $\bc$ can be chosen as
  \[
    [\ba,\bb,\bc] = \left[
      \begin{array}[H]{ccc}
        \sin\varphi_0&\cos\theta_0\cos\varphi_0&\sin\theta_0\cos\varphi_0\\
        -\cos\varphi_0&\cos\theta_0\sin\varphi_0&\sin\theta_0\sin\varphi_0\\
        0&-\sin\theta_0&\cos\theta_0
      \end{array}
    \right].
  \]
\end{remark}

\paragraph{Interpolation}
In \Cref{sec:Gram-Schmidt_orthogonalization}, a Gram-Schmidt
orthogonalization is adopted to evaluate the moment closure. However,
the cost of the orthogonalization is $O(N^6)$, which is a difficult
issue in the numerical simulations. In \cite{MPN}, thanks to the linear
dependence of the moment closure on $E_k$, $2\leq k\leq N$, an
interpolation is applied to reduce the cost. In the 3D \MPN model, the
dependence of the moment closure $\bm{E}_{N+1}$ upon $\bm{E}_{k}$,
$2\leq k\leq N$ is linear. Precisely, we have
\[
  \bm{E}_{N+1} = \sum_{k=0}^{N} {\bf C}_k(\bc_0) \bm{E}_{k},
\]
where ${\bf C}_k(\bc_0)$ is a $(2N+3)\times(2k+1)$ matrix, which
depends on $\bc_0$ only. Therefore, a naive idea is to divide the
unit sphere $\bbS^2$ into several parts, and apply the interpolation.
However, in order to get a sufficient accuracy, dividing the unit
sphere will cost a lot, thus we need another interpolation method.

Noticing the rotational invariance of the 3D \MPN model and the \HMPN
model, we can calculate the moment closure $\overline{\bm{E}_{N+1}}$
corresponding to $\overline{\bc_0}={\bf G}_1\bc_0$, satisfying that
$\overline{\bc_0}$ is parallel to $\be_z$. The procedure of the
evaluation of the moment closure $\bm{E}_{N+1}$ can be written as
\begin{enumerate}
  \item Calculate $\bm{c}_0$, $\overline{\bc_0}$.
  \item Calculate the matrix ${\bf G}_k$, $1\leq k\leq N+1$, 
    whose cost is $O(N^3)$.
  \item Calculate $\overline{\bm{E}}_{k}$, $0\leq k\leq N$ by 
    \[
      \overline{\bm{E}}_k = {\bf G}_k \bm{E}_k,
    \]
    and the cost is $O(N^3)$.
  \item Calculate $\overline{\bm{E}_{N+1}}$ by interpolation, the cost
    is $O(N^3)$.
  \item Calculate $\bm{E}_{N+1}$ by 
    \[
      \bm{E}_{N+1} = {\bf G}_{N+1}^{-1} \overline{\bm{E}_{N+1}}, 
    \]
    the cost is $O(N^2)$.
\end{enumerate}
Therefore, by this interpolation, we reduce the cost of the evolution
of the moment closure to $O(N^3)$, which broadens the application of
the \MPN model.


\section{Hyperbolic Regularization} \label{sec:hyper} 

Similar as the 1D \MPN model, the 3D \MPN model needs to be regularized
to achieve global hyperbolicity. We follow the idea in \cite{HMPN} to
propose a novel hyperbolic regularization for the 3D \MPN model.

At first, we introduce the orthogonal projection to the function space
$\spaceH_N$ in \eqref{eq:spaceH},
\begin{equation}
  \label{eq:orthogonalprojection}
  \mP_N: f = \sum_{\alpha\in\mI} f_{\alpha}\pl_{\alpha} \longrightarrow
  \mP_N f = \sum_{\alpha\in\mI_{N}} f_\alpha \pl_{\alpha}.
\end{equation}
Then the 3D \MPN moment system can be formulated as
\begin{equation}
  \label{eq:MPNprojectionform}
  \dfrac{1}{c}\mP_N\pd{\mP_N I}{t} + \sum_{d=1}^{3} \mP_N\Omega_d
  \pd{\mP_N I}{x_d} = \mP_N \mS(\mP_N I). 
\end{equation}
Using the similar method in \cite{HMPN}, we introduce a new weight function 
\begin{equation}
  \label{eq:weightn}
  \weightn(\bOmega) = \dfrac{1}{(1+\bc_0\cdot\bOmega)^5},
\end{equation}
and the relationship between the weight function $\weight$ and the new
weight function $\weightn$ is 
\[  
   \weight(\bOmega) = (1+\bc_0\cdot\bOmega)\weightn(\bOmega),\quad 
   \nabla_{\bc_0}\weight(\bOmega) = -4\bOmega\weightn(\bOmega) 
\]
Based on the new weight function, we can also define the function
space $\spaceHn_N$, the quasi-orthogonal polynomials $\pln_{\alpha}$,
$\alpha\in\mI_{N}$, the quasi-orthogonal functions $\Pln_{\alpha}$,
the coefficients $\mKn_{\alpha,\beta}$, and the orthogonal projection
$\mPn_N$.

Then, according to the hyperbolic regularization used in \cite{HMPN}, the 3D \HMPN model 
can be written as
\begin{equation}
  \label{eq:HMPNprojectionform}
  \dfrac{1}{c}\mPn_N\pd{\mP_N I}{t} + \sum_{d=1}^{3} \mPn_N\Omega_d
  \mPn_N\pd{\mP_N I}{x_d} = \mPn_N \mS(\mP_N I). 
\end{equation}

Let us show the details of the regularization. In order to calculate
the difference between the \HMPN model \eqref{eq:HMPNprojectionform}
and the \MPN model \eqref{eq:MPNprojectionform}, here we first
introduce a lemma.
\begin{lemma} 
  \label{lem:hyper1}
  \[ 
    \Pl_{\alpha} \in \spaceH_N \subset \spaceHn_{N+1},\quad
    \pd{\Pl_{\alpha}}{c_{0,d}} \in \spaceHn_{N+1},\quad |\alpha|\leq
    N,\ d=1,2,3.
  \]
\end{lemma}
\begin{proof}
  In this proof, we only consider the situation when $|\alpha|=N$. We
  have
  \[ 
    \Pl_{\alpha} = \weight \pl_{\alpha} = \weightn \pl_{\alpha}
    (1+\bc_0\cdot \bOmega),
  \]
  \[
    \pd{\Pl_{\alpha}}{c_{0,d}} = \pd{\weight}{c_{0,d}}\pl_{\alpha}
    +\weight \pd{\pl_{\alpha}}{c_{0,d}} = -4\Omega_d \weightn
    \pl_{\alpha} + (1+\bc_0\cdot \bm\Omega)\weightn
    \pd{\pl_{\alpha}}{c_{0,d}}.
  \]
  We have $\weightn \pl_{\alpha} \in \spaceHn_N$ and
  $(1+\bc_0\cdot\bm\Omega)\weightn\pd{\pl_{\alpha}}{c_{0,d}} \in
  \spaceHn_N$, thus
  \[
    \Pl_{\alpha}-\mPn_N\Pl_{\alpha} =
    (I-\mPn_N)(\weightn\bc_0\cdot\bOmega\pl_{\alpha}) =\sum_{d=1}^{3}
    c_{0,d}\Pln_{\alpha+e_d},
  \]
  \[
    \pd{\Pl_{\alpha}}{c_{0,d}} - \mPn_N \pd{\Pl_{\alpha}}{c_{0,d}} =
    -4\Pln_{\alpha+e_d}.
  \]
\end{proof}
According to the proof of \Cref{lem:hyper1}, we have 
\begin{corollary}
  \label{cor:diffterm}
  \[ 
    \Pl_{\alpha}- \mPn_N{\Pl_{\alpha}} = \left\{
      \begin{aligned}
        &0,\quad &|\alpha|<N,\\
        &\sum_{d=1}^{3} c_{0,d}\Pln_{\alpha+e_d},\quad &|\alpha|=N.
      \end{aligned}
    \right.,
  \]
  \[ 
    \pd{\Pl_{\alpha}}{c_{0,d}}- \mPn_N{\pd{\Pl_{\alpha}}{c_{0,d}}} =
    \left\{
      \begin{aligned}
        &0,\quad &|\alpha|<N,\\   
        &-4\Pln_{\alpha+e_d} ,\quad &|\alpha|=N.
      \end{aligned}
    \right.,
  \]
\end{corollary}
Substituting the terms in \Cref{cor:diffterm} to
\eqref{eq:HMPNprojectionform}, the difference between the \HMPN system
and the \MPN system is
\begin{equation}
  \label{eq:diffHMPn}
  \begin{aligned}
    &\sum_{d=1}^{3} \mPn_N \Omega_d \left(\mPn_N\pd{\mP_N I}{x_d}
      -\pd{\mP_N I}{x_d} \right)\\
    =& \sum_{d=1}^{3} \mPn_N \Omega_d \left( \mPn_N \pd{\left(
          \sum_{\alpha\in\mI_{N}} \delta_{|\alpha|,N}
          f_{\alpha}\Pl_{\alpha} \right)}{x_d} -\pd{\left(
          \sum_{\alpha\in\mI_{N}}
          \delta_{|\alpha|,N}f_{\alpha}\Pl_{\alpha} \right)}{x_d}
    \right) \\
    =& \sum_{d=1}^{3} \mPn_N \Omega_d \left( \sum_{\alpha\in\mI_{N}}
      \delta_{|\alpha|,N}\left( \sum_{l=1}^{3} 4f_{\alpha}
        \Pln_{\alpha+e_l} \pd{c_{0,l}}{x_d}-
        \pd{f_{\alpha}}{x_d}\sum_{l=1}^{3} c_{0,l}\Pln_{\alpha+e_l}
      \right) \right),
  \end{aligned}
\end{equation}
where $\delta_{i,j}$ is the Kronecker delta. Then one can consider
the inner product of \eqref{eq:diffHMPn} with $\bm\Omega^{\alpha}$ to
carry out the \HMPN system, written as
\begin{equation}
  \label{eq:HMPNsystem}
  \dfrac{1}{c} \pd{E_{\alpha}}{t} + 
  \sum_{d=1}^{3} \left(\pd{E_{\alpha+e_d}}{x_d} -
    R_{\alpha,d}\right)
  =S_{\alpha},
\end{equation}
where $S_\alpha \triangleq \langle \mS(\mP_N I)\rangle_{\alpha}$ is
the $\alpha$-th order moment of the right hand side $\mS(\mP_N I)$,
and $R_{\alpha,d}$, according to \eqref{eq:diffHMPn}, is
\begin{equation}
  \label{eq:termR}
  R_{\alpha,d}=
  \begin{cases}
    0, \quad &|\alpha|<N,\\
    \sum\limits_{\beta\in\mI_{N},|\beta|=N}
    \sum\limits_{l=1}^{3}\left(\pd{f_{\beta}}{x_d}c_{0,l}-4f_{\beta}\pd{c_{0,l}}{x_d}\right)
    \mKn_{\alpha+e_d,\beta+e_l} , \quad &|\alpha|=N.
  \end{cases}
\end{equation}
\begin{remark}
  If $\alpha+e_d \notin \mI$, i.e. $\alpha_3 = 1$ and $d=3$, then
  $\mKn_{\alpha+e_d,\beta+e_l}$ can be calculated by 
  \[
  \mKn_{\alpha+e_3, \beta+e_l} = \mKn_{\alpha-e_3, \beta+e_l}
  -\mKn_{\alpha-e_3+2e_1, \beta+e_l} - \mKn_{\alpha-e_3+2e_2, \beta+e_l}.
  \]
  Analogously, for $\beta+e_l\notin \mI$,
  $\mKn_{\alpha+e_d,\beta+e_l}$ can also be calculated.
\end{remark}
Therefore, one can conclude that
\begin{property}
  The regularization does not change the moment equations for
  $E_{\alpha}$, $\alpha\in\mI_{N-1}$.
\end{property}
On the other hand, notice that if $N=1$, according to \eqref{eq:termR}
and \eqref{eq:f1}, we have that the \HMPN model is exact the same as
the \MPN model. Thus, as the \HMPN model in slab geometry \cite{HMPN},
we have the following property,
\begin{property}
  The hyperbolic regularization in 3D case does not change the \Mone
  model, in other words, the \HMPN model is the same as the \Mone
  model when $N=1$.
\end{property}

In \Cref{sec:model}, the rotational invariance of the 3D \MPN model
has been proved and in this section we investigate the rotational
invariance of the 3D \HMPN model. According to
\Cref{thm:rotationalinvariance}, the remaining part is $R_{\alpha,d}$
in \eqref{eq:termR}, a simple conclusion is given as
\begin{theorem}
  The 3D \HMPN model is rotational invariant.
\end{theorem}
\begin{proof}
  According to \eqref{eq:RI_E}, one only needs to prove that  
  \begin{equation}
    \label{eq:HMPN_RI_target0}
    \sum_{d=1}^{3} \overline{R_{\alpha, d}} = 
    \sum_{d=1}^{3}\sum_{|\beta| = N,
      \beta\in\mI}\bG_{N,\mathcal{N}(\alpha,N),\mathcal(\beta,N)}R_{\beta,
      d},
    \quad |\alpha|=N, \alpha\in\mI.
  \end{equation}
  According to \eqref{eq:termR}, we have 
  \begin{equation}
    \begin{aligned}
      \label{eq:HMPN_RI_target}
      \sum_{d=1}^{3} \overline{R_{\alpha,d}} = \sum_{d=1}^{3}
      \sum\limits_{|\gamma|=N, \gamma\in\mI}
      \sum\limits_{l=1}^{3}\left(\pd{\overline{f_{\gamma}}}
        {\overline{x_d}}\overline{c_{0,l}}-4\overline{f_{\gamma}}
        \pd{\overline{c_{0,l}}}{\overline{x_d}}\right)
      \overline{\mKn_{\alpha+e_d,\gamma+e_l}} , \quad &|\alpha|=N.
    \end{aligned}
  \end{equation}
  Noticing that 
  \[
    \begin{aligned}
      &\overline{f_{\gamma}} = \sum_{|\gamma'|=N, \gamma'\in\mI}
      \bG_{N, \mathcal{N}(\gamma, N), \mathcal{N}(\gamma',N)}
      f_{\gamma'}, \quad \overline{x_d} =
      \sum_{d'=1}^{3}\bG_{1,d,d'}x_d', \quad
      \overline{c_{0,l}} = \sum_{l'=1}^{3}\bG_{1,l,l'}c_{0,l'},\\
      &\overline{\mKn_{\alpha+e_d,\gamma+e_l}} = \sum_{d'',l''=1}^{3}
      \sum_{|\alpha''|=N,\alpha''\in\mI}
      \sum_{|\gamma''|=N,\gamma''\in\mI} \bG_{1,d,d''}
      \bG_{N,\mathcal{N}(\alpha,N),\mathcal{N}(\alpha'',N)}
      \mKn_{\alpha''+e_{d''},\gamma''+e_{l''}}
      \bG^T_{N,\mathcal{N}(\gamma'',N),\mathcal{N}(\gamma,N)}
      \bG^T_{1,l'',l},
    \end{aligned}
  \]
  direct calculations yield that \eqref{eq:HMPN_RI_target} can be
  simplified as
  \[
    \sum_{d=1}^{3} \overline{R_{\alpha,d}} = \sum_{d=1}^{3}
    \sum_{|\alpha''|=N,\alpha''\in\mI}\bG_{N,
      \mathcal{N}(\alpha,N),\mathcal{N}(\alpha'',N)} R_{\alpha'',d},
  \]
  which is equivalent to \eqref{eq:HMPN_RI_target0}.
\end{proof}

According to the calculations before, the ansatz $\hat{I}$ can be
determined by the moments $E_{\alpha}$, $\alpha\in\mI_{N}$.  According
to \eqref{eq:ansatz}, meanwhile, the ansatz can be also determined by
$c_{0,1}, c_{0,2}, c_{0,3}, f_{\alpha}$, $\alpha \in \mI_{N}$. Noticing
that $f_{e_1}=f_{e_2}=f_{e_3}=0$, we can use a vector
$\bw\in\bbR^{(N+1)^2}$, which is defined as
\begin{equation}\label{eq:bwdefine}
  \bw = \left( f_0, c_{0,1}, c_{0,2}, c_{0,3}, f_{\alpha},
    \alpha\in\mI_{N} , |\alpha|\geq 2 \right)^T,
  \end{equation}
to describe the ansatz. Therefore, we can rewrite the \MPN system
\eqref{eq:momentsystem} and the \HMPN system \eqref{eq:HMPNsystem} by
using the variables $\bw$ instead of $E_{\alpha}$.

Notice that $\mPn_N\pd{\mP_N I}{s}$, $s=t,x_1,x_2,x_3$, can be linearly
expressed by $\Pln_{\alpha}$, $\alpha\in\mI_{N}$, thus there exist a
$(N+1)^2\times (N+1)^2$ matrix $\bD$, satisfies that
\begin{equation}
  \label{eq:matrixDdeduce}
  \mPn_N \pd{\mP_N I}{s} = (\bm\Pln)^T \bD \pd{\bw}{s}, \quad s = t,x_1,x_2,x_3, 
\end{equation}
where $\bm\Pln$ is a vector of all quasi-orthogonal functions 
$\Pln_{\alpha}$ for $\alpha\in\mI_N$, whose dimension is $(N+1)^2$.
Therefore, the \HMPN system \eqref{eq:HMPNprojectionform} can be
rewritten as
\begin{equation}
  \label{eq:HMPNsystem_matrix_1}
  \dfrac{1}{c} (\bm\Pln)^T \bD \pd{\bw}{t}  
  + \sum_{d=1}^{3} \mPn_N \Omega_d (\bm\Pln)^T \bD \pd{\bw}{x_d} = S.
\end{equation}
Taking the inner product of \eqref{eq:HMPNsystem_matrix_1} with
$\bm\Pln$, one can obtain the matrix form of the \HMPN system, written
as
\begin{equation}
  \label{eq:HMPNsystem_matrix}
  \dfrac{1}{c} \bm\Lambda \bD \pd{\bw}{t}  
  + \sum_{d=1}^{3} \bM_d \bD \pd{\bw}{x_d} = \tilde{\bm S},
\end{equation}
where $\bLambda_{\alpha,\beta} = \inner{\Pln_{\alpha}}{\Pln_{\beta}}$
is symmetric positive definite, and
$\bM_{d,\alpha,\beta} = \inner{\Pln_{\alpha}}{\Omega_d\Pln_{\beta}}$
are symmetric, for $d=1,2,3$, and
$\tilde{S}_{\alpha} = \int_{\bbS^2} S \Pln_{\alpha} \dd\bOmega$.

\begin{theorem}
  The 3D \HMPN model is globally hyperbolic. Precisely,
  $\forall \bm{n} \in \bbS^2$,
  $\bA\triangleq c\sum_{d=1}^{3} (\bLambda\bD)^{-1}(n_d\bM_d\bD)$ is real
  diagonalizable. Moreover, the eigenvalue of $\bA$ is not greater
  than the speed of light.
\end{theorem}
\begin{proof}
  Notice that, for $\forall \bm{n}\in\bbS^2$,
  \[ 
    \bA = c\sum_{d=1}^{3} n_d\bD^{-1}\bLambda^{-1}\bM_d\bD =
    \bD^{-1}\bLambda^{-1}\bM\bD,
  \]
  where $\bLambda$ is symmetric positive definite, and
  $\bM_{\alpha,\beta} = \innern{\Pln_{\alpha}}
  {\sum_{d=1}^{3}n_d\Omega_d\Pln_{\beta}}$ is symmetric, we have that
  $\bA$ is diagonalisable with real eigenvalues, thus the \HMPN system
  is globally hyperbolic.

  Moreover, denote $\lambda$ as the eigenvalue of $\bA$, then
  $\lambda$ is the eigenvalue of $c\bLambda^{-1}\bM$, and suppose
  $\bxi\in\bbR^{(N+1)^2}$ is the corresponding eigenvector, i.e.
  $ \bM \bxi = \lambda \bLambda\bxi$. Precisely, denote
  $\mathcal{F}=\bxi^T\bm\Pl\in\spaceHn_N$, we have
  \begin{equation}
    \label{eq:eigenvaluededuce}
    \innern{\Pln_{\alpha}}{\left(
        c\sum_{d=1}^{3}n_d\Omega_d-\lambda\right)\mathcal{F}} = 0,\quad \forall
    \alpha\in\mI_{N}.
  \end{equation}
  Notice $\mathcal{F}\in\spaceHn_N$, thus \eqref{eq:eigenvaluededuce}
  implies 
  \[   
    \innern{\mathcal{F}}{\left(
        c\sum_{d=1}^{3}n_d\Omega_d-\lambda\right)\mathcal{F}} = 0.
  \]
  Therefore, 
  \[ 
    \lambda =
    \dfrac{c\innern{\mathcal{F}}{(\bm{n}\cdot\bOmega)\mathcal{F}}}
    {\innern{\mathcal{F}}{\mathcal{F}}}.
  \]
  Notice $|\bm{n}\cdot\bOmega|\leq 1$, we have 
  \[
    |\lambda|\leq c.
  \]
\end{proof}

At the end of this section, we give a summary of the properties of the
3D \HMPN model.
\begin{itemize}
\item The 3D \HMPN model is rotational invariant.
\item The 3D \HMPN model is globally hyperbolic.
\item All the characteristic speeds of the 3D \HMPN model are not
  greater than the speed of light.
\item The hyperbolic regularization adopted in the 3D \HMPN model does
  not change the governing equations of $E_{\alpha}$ for
  $|\alpha|\leq N-1$ in the 3D \MPN model.
\item The hyperbolic regularization vanishes when $N=1$, i.e. the 3D
  \HMPN model is exactly the same as the 3D \MPN model and the 3D \Mone
  model.
\end{itemize}


\section{Numerical Validation} \label{sec:num}

In this section, we try to validate the 3D \HMPN model by numerical
experiments. For this purpose, we first give a preliminary numerical
scheme for the regularized reduced model \eqref{eq:HMPN}, and then
perform numerical simulations on some typical examples to demonstrate
its validity. Due to practical difficulties in a 3D computation and
noticing that our aim is to validate the new model, we currently only
perform our numerical simulations on 2D domains. Actually, 2D numerical
results are enough to demonstrate the major features of our model.

\subsection{Numerical scheme} 
We first collect the equations in the regularized reduced model
\eqref{eq:HMPNsystem} in a matrix formation as
\begin{equation}\label{eq:HMPN}
  \dfrac{1}{c} \pd{\bU}{t} + \sum_{d=1}^{2} \left[
    \pd{\bF_d(\bU)}{x_d} + {\bf R}_d(\bU)\pd{\bU}{x_d}\right] = \bm S,
\end{equation}
where $\bU$, $\bF(\bU)$, ${\bf R}_d \pd{\bU}{x_d}$ and $\bm S$ are
vectors in $\mathbb{R}^{(N+1)^2}$, whose $\alpha$-th element are
$\bU_{\alpha} = E_{\alpha}$, $\bF_{d,\alpha} = E_{\alpha+e_d}$,
$\left({\bf R}_d \pd{\bU}{x_d}\right)_{\alpha} = R_{\alpha, d}$, and
$\bm S_\alpha = S_\alpha$, respectively.

Denote the computational domain by $[x_l, x_r]\times[y_l,y_r]$, where
we can make a finite volume discretization conveniently. The
domain is partitioned uniformly into $N_x\times N_y$ cells.  The
$(i,j)$-th mesh cell is
$[x_{i-1/2}, x_{i+1/2}] \times [y_{j-1/2},y_{j+1/2}]$,
$i = 1, \dots, N_x$, $j = 1, \dots, N_y$ with
$x_{i+1/2}=x_l+i\Delta x$, $y_{j+1/2} = y_l+j\Delta y$, and
$\Delta x=\frac{x_r - x_l}{N_x}$, $\Delta y = \frac{y_r -
  y_l}{N_y}$. Let $\bU_{i,j}^n$ be the approximation of the solution
$\bU$ on the $(i,j)$-th mesh cell at the $n$-th time step $t_n$.

For the purpose to validate the model, the numerical scheme for
\eqref{eq:HMPN} the simpler the better that we adopt abruptly a
splitting scheme in each time step. Precisely, we split it into three
parts: convection terms in two spatial directions and the source term
as
\begin{align}
  \label{eq:convection}
  \text{convection in }x\text{-direction:}\quad
  & \frac{1}{c}\pd{\bU}{t}+
    \pd{\bF_1(\bU)}{x}+\bR_1(\bU)\pd{\bU}{x}=0,   \\
  \text{convection in }y\text{-direction:}\quad
  & \frac{1}{c}\pd{\bU}{t}+
    \pd{\bF_2(\bU)}{y}+\bR_2(\bU)\pd{\bU}{y}=0,   \\
    \label{eq:source_num}
  \text{source part:}\quad & \frac{1}{c}\pd{\bU}{t} = \bm S.
\end{align}
Next, we give the numerical scheme for both parts in our code.

\paragraph{Source term}
The right hand side $\mS(I)$ denotes the actions by the background
medium on the photons. Generally, it contains a scattering term, an
absorption term, and an emission term, and has the form
\cite{Bru02,McClarren2008Semi}
\begin{equation} \label{eq:sourceterm_initial} \mS(I) =
  \dfrac{1}{4\pi}\sigma_s\int_{\bbS^2} I \dd\mu -(\sigma_a + \sigma_s)
  I+ \dfrac{1}{4\pi}\sigma_a acT^4 + \dfrac{s}{4\pi},
\end{equation}
where $a$ is the radiation constant; $T(\bx,t)$ is the material
temperature; $\sigma_a(\bx,T)$, $\sigma_s(\bx,T)$ and
$\sigma_t=\sigma_a+\sigma_s$ are the absorption, scattering, and total
opacity coefficients, respectively; and $s(\bx)$ is an isotropic
external source. The temperature is related to the internal energy
$e$, whose evolution equation is
\begin{equation}\label{eq:internalenergy}
  \pd{e}{t} = \sigma_a \left(\int_{\bbS^2} I\dd\bm\Omega -acT^4\right).
\end{equation}
The relationship between $T$ and $e$ is problem-dependent, and we will
assign it in the numerical examples when necessary.

Noticing the quartic term $a c\sigma_a T^4$ in $\mS(I)$ and the
evolution equation of $e$ \eqref{eq:internalenergy}, we adopt the
implicit Euler scheme on them as
\[
  \frac{\bU_{i,j}^{n+1}-\bU_{i,j}^n}{c\Delta t} = \bm S^{n+1}_{i,j},\quad
  \dfrac{e_{i,j}^{n+1}-e_{i,j}^n}{\Delta t} = \sigma_{a,i,j}^{n+1} 
  \left( E_{0,i,j}^{n+1}-ac(T_{i,j}^{n+1})^4 \right).
\]
One can directly check that in the absence of any external source of
radiation, i.e., $s = 0$, this discretization satisfies the
conservation of total energy as
\[
  \frac{e^{n+1}_{i,j}-e^n_{i,j}}{\Delta t} +
  \frac{E^{n+1}_{0,i,j}-E^{n}_{0,i,j}}{c\Delta t} = 0.
\]

\paragraph{Convection part}
Without loss of generality, we only consider the convection part in
$x$-direction. For a better reading experience below, the subscript of
dimension is suppressed, i.e. we will use $\bF$ and $\bR$ instead of
$\bF_1$ and $\bR_1$. The hyperbolic regularization in \cref{sec:hyper}
modifies the governing equation of $E_{\alpha}$, $\alpha\in\mI_{N}$,
$|\alpha|=N$, such that these equations can not be written into a
conservation form.  Therefore, the classical Riemann solvers for
hyperbolic conservation laws can not be directly applied to solve
\eqref{eq:convection}. In the numerical simulations of the \HMPN
system in slab geometry \cite{HMPN}, we adopt the DLM theory
\cite{Maso} to deal with the non-conservation terms, and we continue
to use this method here. Precisely, the key is introducing a path
$\Gamma(\tau;\cdot,\cdot)$, $\tau\in[0,1]$ to connect two states
$\bU^L$ and $\bU^R$ beside the Riemann problem such that
\[
  \Gamma(0;\bU^L,\bU^R) = \bU^L,\quad
  \Gamma(1;\bU^L,\bU^R) = \bU^R.
\]
The path allows a generalization of the Rankine-Hugoniot condition to
the non-conservation system as
\begin{equation} \label{eq:generalizedRH} \bF(\bU^L) - \bF(\bU^R) +
  \int_0^1 [v_s{\bf I}-\bR(\Gamma(\tau;\bU^L,\bU^R))]
  \pd{\Gamma}{\tau}(\tau;\bU^L,\bU^R)\dd \tau = 0,
\end{equation}
if the two states $\bU^L$ and $\bU^R$ are connected by a shock with
shock speed $v_s$.  Then the weak solution of the non-conservation
system can be defined.  Readers can find more details of the
constrained path and the theory results in \cite{Maso}.  We then
introduce the finite volume scheme in \cite{Rhebergen} to discretize
the non-conservation system \eqref{eq:convection}. This scheme can be
treated as a non-conservation version of the HLL scheme and has been
successfully applied to some non-conservation models
\cite{Microflows1D, Qiao}.

Applying the finite volume scheme in \cite{Rhebergen} yields
\begin{equation}
  \label{eq:numericalscheme}
  \dfrac{\bU^{n+1}_{i,j}-\bU_{i,j}^n}{c\Delta t} + \dfrac{
    \hat{\bF}_{i+1/2,j}^{n}-\hat{\bF}_{i-1/2,j}^n}{\Delta x}
  +\dfrac{\hat{\bR}_{i+1/2,j}^{n-}-\hat{\bR}_{i-1/2,j}^{n+}}{\Delta x} = 0.
\end{equation} 
Here the flux $\hat{\bF}^n_{i+1/2,j}$ is the HLL numerical flux for
the conservation term $\pd{\bF(\bU)}{x}$, given by
\begin{equation} \label{eq:HLLflux_con}
  \hat{\bF}_{i+1/2,j}^{n} =
  \begin{cases}
    \bF(\bU_{i,j}^n), &  \lambda^L_{i+1/2,j}\geq 0,\\
    \dfrac{\lambda^R_{i+1/2,j}
      \bF(\bU_{i,j}^n)-\lambda^L_{i+1/2,j}\bF(\bU_{i+1,j}^n) +
      \lambda^L_{i+1/2,j}\lambda^R_{i+1/2,j}(\bU^n_{i+1,j}-\bU^n_{i,j})}
    {\lambda^R_{i+1/2,j}-\lambda^L_{i+1/2,j}}, &
    \lambda^L_{i+1/2,j}<0<\lambda^R_{i+1/2,j},\\
    \bF(\bU_{i+1,j}^n),   &  \lambda^R_{i+1/2,j}\leq 0,\\
  \end{cases}
\end{equation}
where $\lambda^L_{i+1/2,j}$ and $\lambda^R_{i+1/2,j}$ are defined as
\[
  \lambda^L_{i+1/2,j} = \min(\lambda^L_{i,j},\lambda^L_{i+1,j}),\quad
  \lambda^R_{i+1/2,j} = \max(\lambda^R_{i,j},\lambda^R_{i+1,j}).
\]
Here $\lambda^{L}_{i,j}$ and $\lambda^R_{i,j}$ are the minimum and
maximum characteristic speeds of $\bU_{i,j}^{n}$, respectively.  
flux $\hat{\bR}_{i+1/2,j}^{n\pm}$ is the special treatment of the
finite volume scheme in \cite{Rhebergen} for the non-conservation term
$\bR(\bU)\pd{\bU}{x}$, given by
\begin{equation} \label{eq:HLLflux_nonconm}
  \hat{\bR}_{i+1/2,j}^{n-}=
  \begin{cases}
    0,  &  \lambda^L_{i+1/2,j}\geq 0,\\
    -\dfrac{\lambda^L_{i+1/2,j}\bg_{i+1/2,j}^n}
    {\lambda^R_{i+1/2,j}-\lambda^L_{i+1/2,j}},  &
    \lambda^L_{i+1/2,j}<0<\lambda^R_{i+1/2,j},\\
    \bg_{i+1/2,j}^n,   &  \lambda^R_{i+1/2,j}\leq 0,\\
  \end{cases}
\end{equation}
and 
\begin{equation} \label{eq:HLLflux_nonconp}
  \hat{\bR}_{i+1/2,j}^{n+}= 
  \begin{cases}
    -\bg^{n}_{i+1/2,j},    &  \lambda^L_{i+1/2,j}\geq 0,\\
    -\dfrac{\lambda^R_{i+1/2,j}\bg_{i+1/2,j}^n}
    {\lambda^R_{i+1/2,j}-\lambda^L_{i+1/2,j}},  &
    \lambda^L_{i+1/2,j}<0<\lambda^R_{i+1/2,j},\\
    0,  &  \lambda^R_{i+1/2,j}\leq 0,\\
  \end{cases}
\end{equation}
where  
\begin{equation} \label{eq:integrateR}
  \bg_{i+1/2,j}^n=\int_{0}^{1}{\bf
    R}(\Gamma(\tau;\bU_{i,j}^n,\bU_{i+1,j}^n))
  \pd{\Gamma}{\tau}(\tau;\bU_{i,j}^n,\bU_{i+1,j}^n)\dd \tau.
\end{equation}

Since the implicit scheme is adopted in the discretization of the
source term, one can easily check that the discretization is
unconditionally stable. Thus the time step is constrained by the
convection term and complies with the CFL condition
\[
  \text{CFL} := \max_{i,j} |\lambda(\bU_{i,j}^n)| \frac{\Delta
  t}{\Delta x} < 1.
\]

Notice $\bm w$ and $\bm U$ are uniquely determined by each other, therefore the
path $\Gamma(\tau;\bm U^L, \bm U^R)$ is equivalent to the path $\gamma(\tau;\bm
w^L, \bm w^R)$, where $\bw$ is defined in \eqref{eq:bwdefine}, and $\bm w^L$ and $\bm w^R$ 
are the value of $\bw$ when $\bU$ is equal to $\bU^L$ and $\bU^R$, respectively.
In \cite{HMPN}, we verified that the choice of the path $\Gamma(\tau;\cdot,\cdot)$
is not essential, thus in this paper we simply choose the path as a linear path
between $\bw^L$ and $\bw^R$, i.e. 
\[
\gamma(\tau;\bw^L,\bw^R) = \bw^L + \tau (\bw^R-\bw^L), \quad 0\leq \tau\leq 1.
\]
\paragraph{Boundary condition}
We adopt the method in \cite{MPN} and \cite{HMPN} to deal with the
boundary condition.  In \cref{sec:model}, one can determine an
injective function between the distribution function $\hat{I}$ and the
moments $E_{\alpha}$, thus we can construct the boundary condition of
the reduced model based on the boundary condition of the RTE. Without
loss of generality, we take the left boundary, i.e. $x=x_l$ and $i=0$
in $U_{i,j}$ as an example.

On the left boundary, the specific intensity is given by
\[
  I^B(t, \bx; \bOmega) = \left\{
  \begin{aligned}
    &I(t, \bx; \bOmega),\quad &\bOmega \cdot \be_{x}<0,\\
    &I_{\text{out}}(t, \bx; \bOmega),\quad &\bOmega\cdot\be_{x}>0,
  \end{aligned}
  \right.
\]
where $I_{\text{out}}$ is the specific intensity outside of the
domain, which depends on the specific problem and the intensity inside
the domain on the boundary $I(t, \bx; \bOmega)$, where $\bx_1 =
x_l$. Here we list some of the commonly used boundary conditions and the
choices of the intensity $I_{\text{out}}$, for later usage.

\begin{itemize}
  \item Infinite boundary condition:
    \[
      I_{\text{out}}(t, \bx; \bOmega) = I(t, \bx; \bOmega), \quad
      \bOmega\cdot\be_{x} > 0.
    \]
  \item Reflective boundary condition:
    \[
      I_{\text{out}}(t, \bx; \bOmega) = I(t, \bx; \bOmega -
      2(\bOmega\cdot\be_x)\be_x), \quad \bOmega\cdot\be_{x} > 0.
    \]
  \item Vacuum boundary condition:
    \[
      I_{\text{out}}(t, \bx; \bOmega) = 0, \quad \bOmega\cdot\be_{x} >
      0.
    \]
  \item Inflow boundary condition:
    \[
      I_{\text{out}}(t, \bx; \bOmega) = I_{\text{inflow}}(t, \bx;
      \bOmega), \quad \bOmega\cdot\be_x > 0.
    \]
    where $I_{\text{inflow}}$ is the specific intensity of the
    external inflow.
\end{itemize}

Furthermore, we replace the intensity $I(t, \bx; \bOmega)$ by the
specific intensity constructed by the moments in the cell near the
left boundary. Precisely,
\[
  I(t, \bx; \bOmega) = \hat{I}\left(\bOmega; \bU(t, \bx)\right),
\]
where $\bU(t, \bx)$ is the moments in the cell near the left boundary,
i.e. $[x_{1/2}=x_l,x_{3/2}=x_l+\Delta x]\times[y_{j-1/2}, y_{j+1/2}]$
at time $t$. Then one can directly obtain the flux across the left
boundary. Precisely, the $\alpha$-th flux at $t$ and $\bx$ is given by
\[
  \int_{\bbS^2}\bOmega^{\alpha+e_1}I^B(t, \bx; \bOmega)\dd\bOmega =
  \int_{\bOmega\cdot\bm{e}_x<0} \bOmega^{\alpha+e_1}
  \hat{I}(\bOmega;\bU(t, \bx))\dd\bOmega
  +\int_{\bOmega\cdot\bm{e}_x>0}
  \bOmega^{\alpha+e_1}I_{\text{out}}(t,\bx; \bOmega)\dd\bOmega.
\]

\subsection{Numerical examples}
Below, we present some numerical examples to show different features of
the 3D \HMPN model.

\paragraph{Inflow problem}
This example is used to study the behaviour of the solution of the
\HMPN model, hence the right hand side vanishes, i.e. the RTE
degenerates into
\begin{equation}
  \dfrac{1}{c}\pd{I}{t} + \bm\Omega \cdot \nabla_{\bx} I = 0.
  \label{eq:RTE_norhs}
\end{equation}
The initial state is chosen as
\[
  I_0(\bx; \bOmega) =\left\{
\begin{aligned}
  &ac\delta(\bOmega-\bOmega_0),\quad & x<0 \text{ and } y<0,\\
  &\dfrac{10^{-3}}{4\pi}ac,\quad &\text{otherwise},
\end{aligned}
\right.
\]
where $\Vert \bOmega_0\Vert=1$. 
The distribution function is a Dirac delta function, 
which is extremely anisotropic and hard to approximate 
with the \PN model, which approximates the specific intensity with polynomials. 
On the other hand, due to the fact the \SN model only considers the
particles along with a discrete set of angular directions, when the
$\bm\Omega_0$ does not belong to this set, the \SN model can not get a
good approximation. 

The computational domain is $[-0.2, 0.2]\times[-0.2, 0.2]$, and the infinite
boundary conditions are prescribed at the boundaries. 

\begin{figure}[htbp]
  \centering 
  \subfloat[$\cos\theta_0 = \dfrac{3}{10}$, $ct_{\text{end}}=0.05$]{
      \includegraphics[width=0.33\textwidth
      ,trim={10mm 8mm 15mm 8mm}, clip
      ]{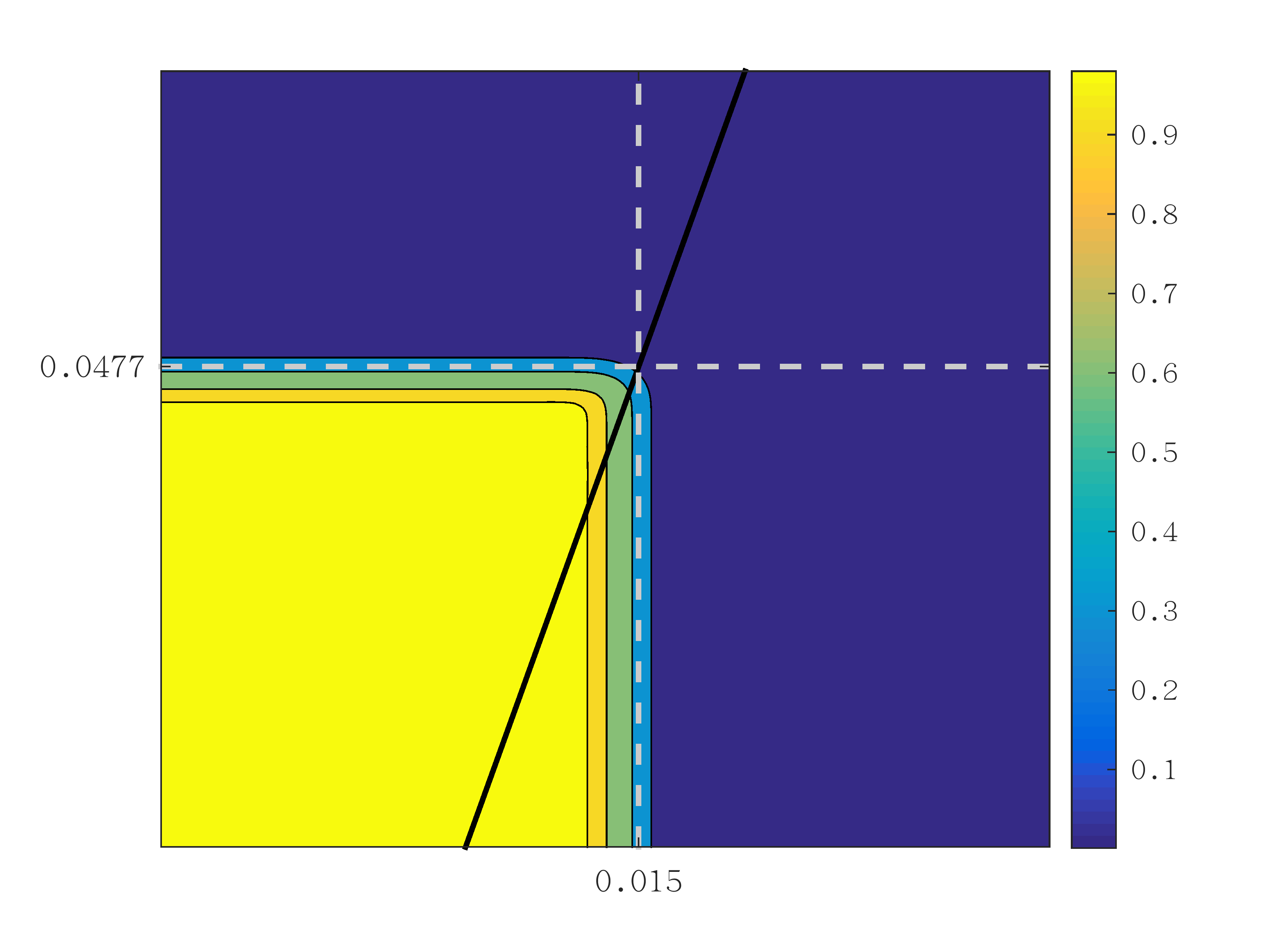}
  }
  \subfloat[$\cos\theta_0 = \dfrac{3}{10}$, $ct_{\text{end}}=0.1$]{
      \includegraphics[width=0.33\textwidth
      ,trim={10mm 8mm 15mm 8mm}, clip
      ]{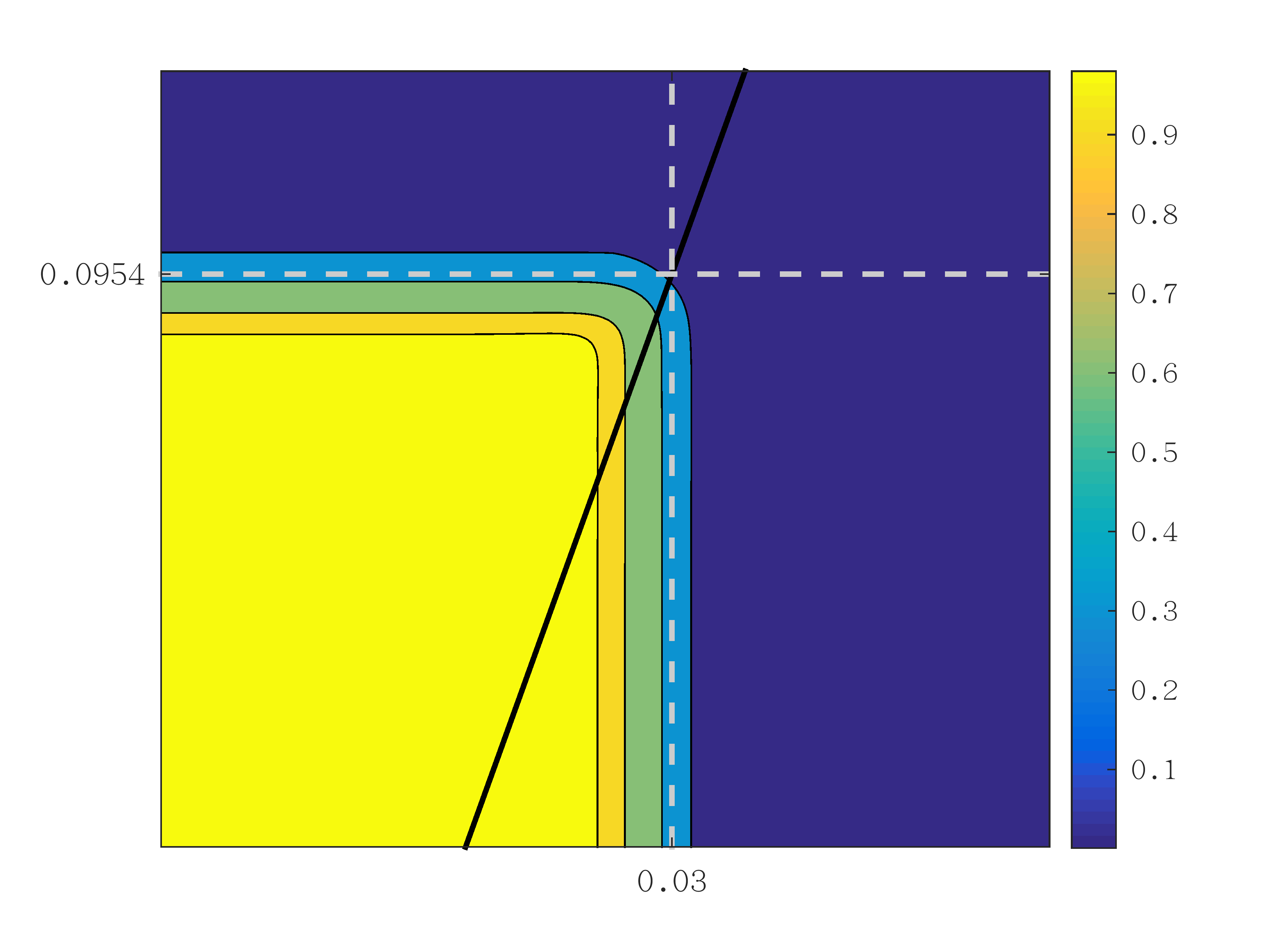}
  }
  \subfloat[$\cos\theta_0 = \dfrac{3}{10}$, $ct_{\text{end}}=0.2$]{
      \includegraphics[width=0.33\textwidth
      ,trim={10mm 8mm 15mm 8mm}, clip
      ]{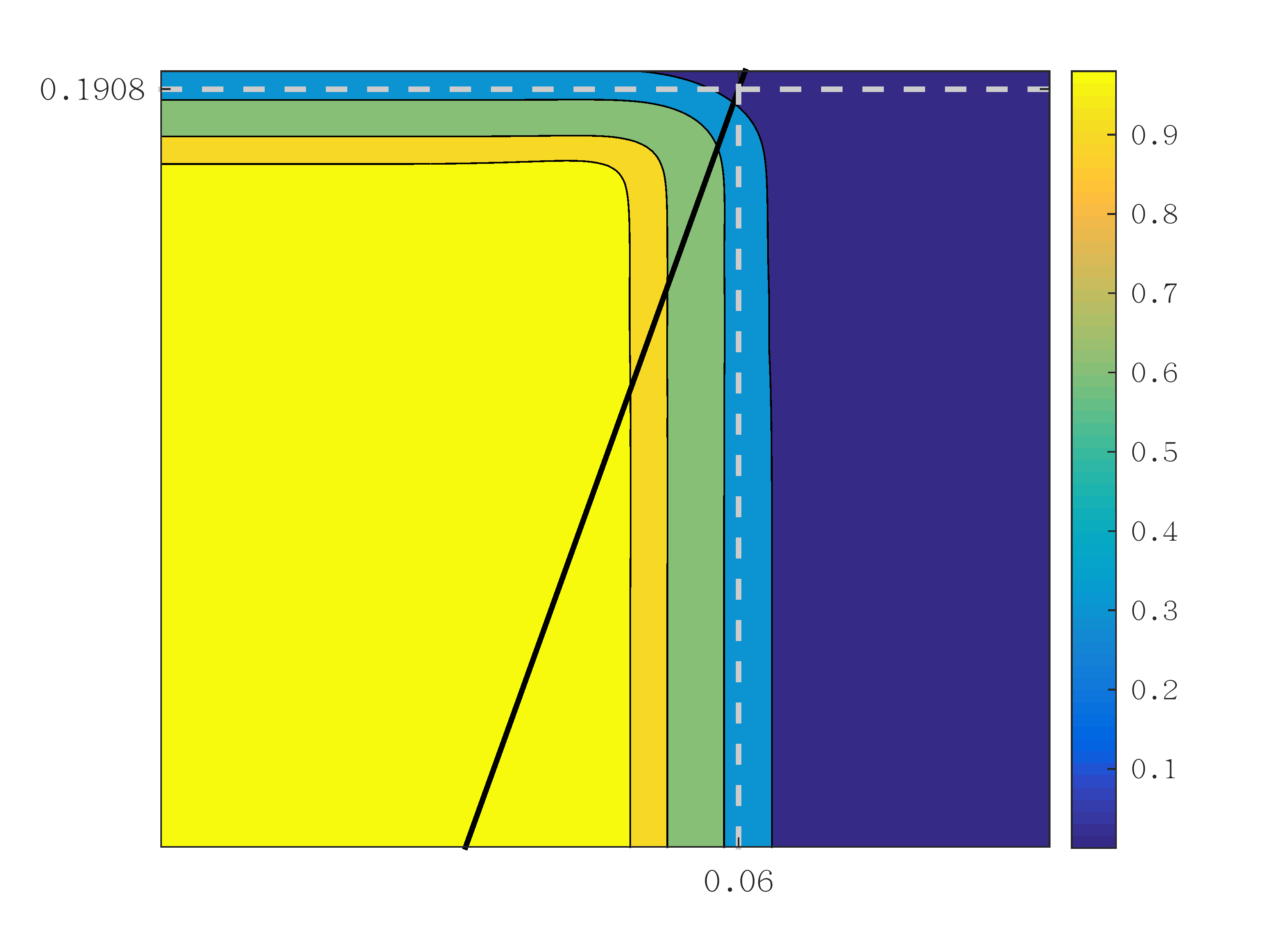}
  }\\
  \subfloat[$\cos\theta_0 = \dfrac{6}{10}$, $ct_{\text{end}}=0.05$]{
      \includegraphics[width=0.33\textwidth
      ,trim={10mm 8mm 15mm 8mm}, clip
      ]{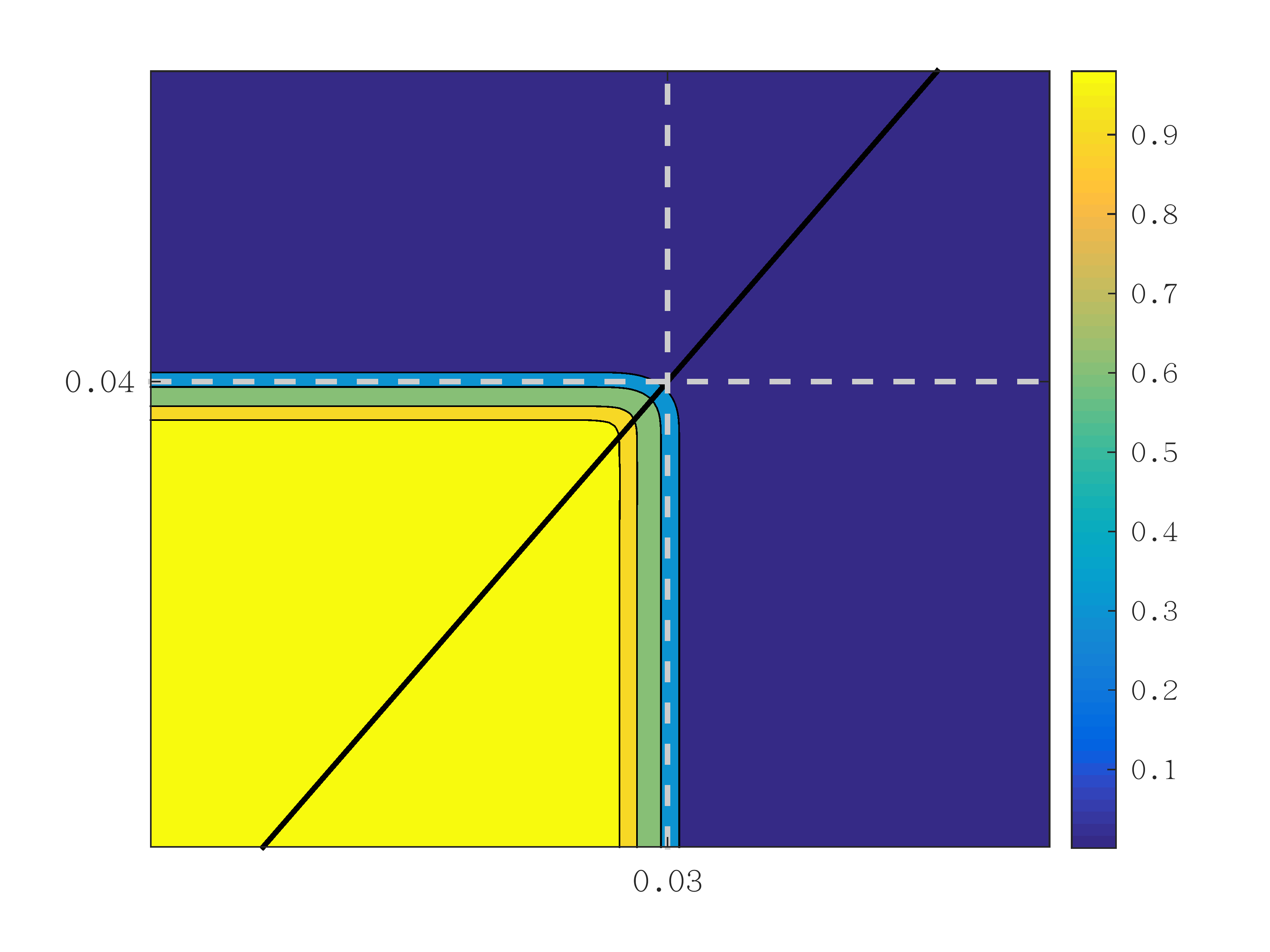}
  }
  \subfloat[$\cos\theta_0 = \dfrac{6}{10}$, $ct_{\text{end}}=0.1$]{
      \includegraphics[width=0.33\textwidth
      ,trim={10mm 8mm 15mm 8mm}, clip
      ]{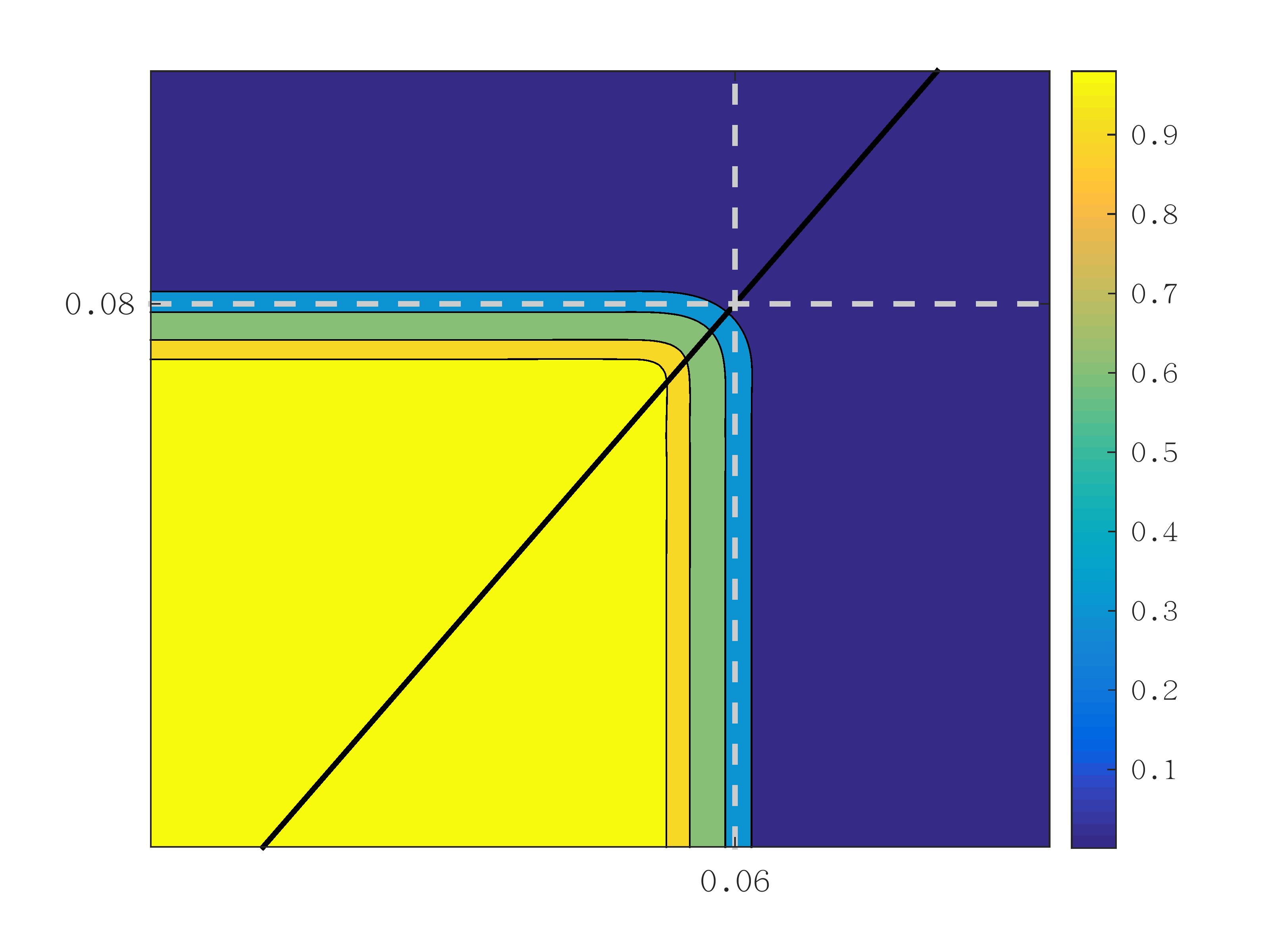}
  }
  \subfloat[$\cos\theta_0 = \dfrac{6}{10}$, $ct_{\text{end}}=0.2$]{
      \includegraphics[width=0.33\textwidth
      ,trim={10mm 8mm 15mm 8mm}, clip
      ]{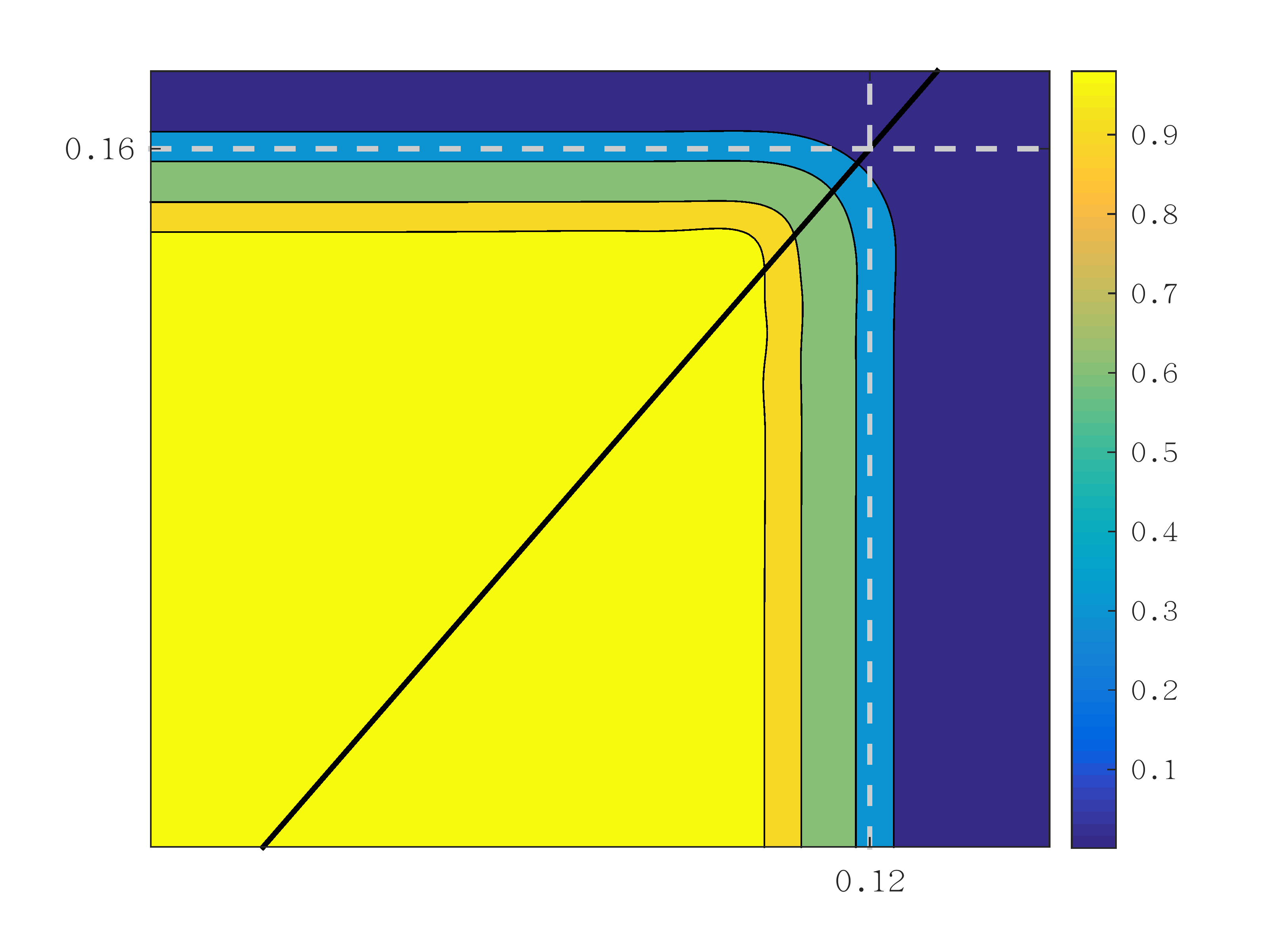}
  }\\
  \subfloat[$\cos\theta_0 = \dfrac{9}{10}$, $ct_{\text{end}}=0.05$]{
      \includegraphics[width=0.33\textwidth
      ,trim={10mm 8mm 15mm 8mm}, clip
      ]{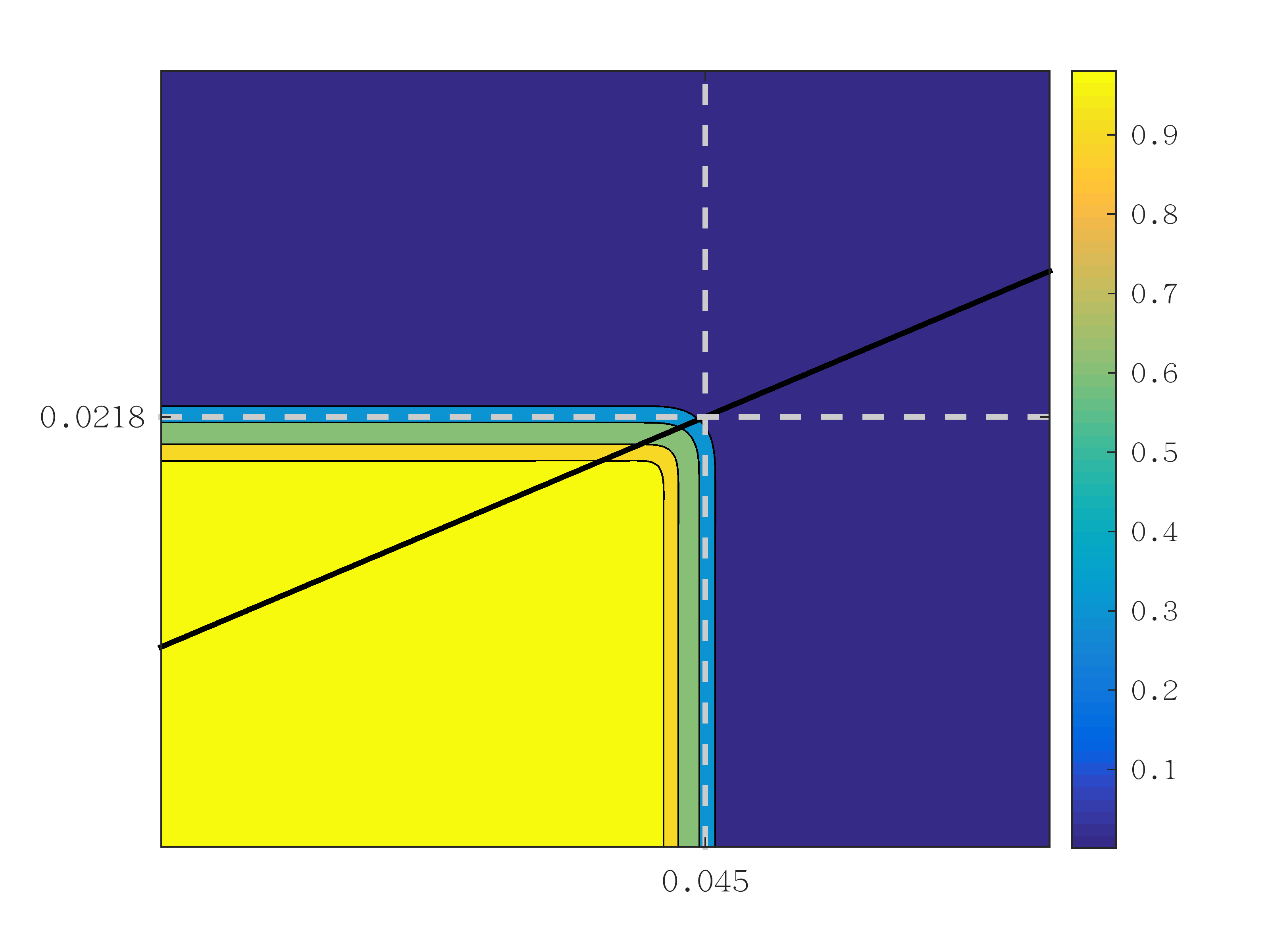}
  }
  \subfloat[$\cos\theta_0 = \dfrac{9}{10}$, $ct_{\text{end}}=0.1$]{
      \includegraphics[width=0.33\textwidth
      ,trim={10mm 8mm 15mm 8mm}, clip
      ]{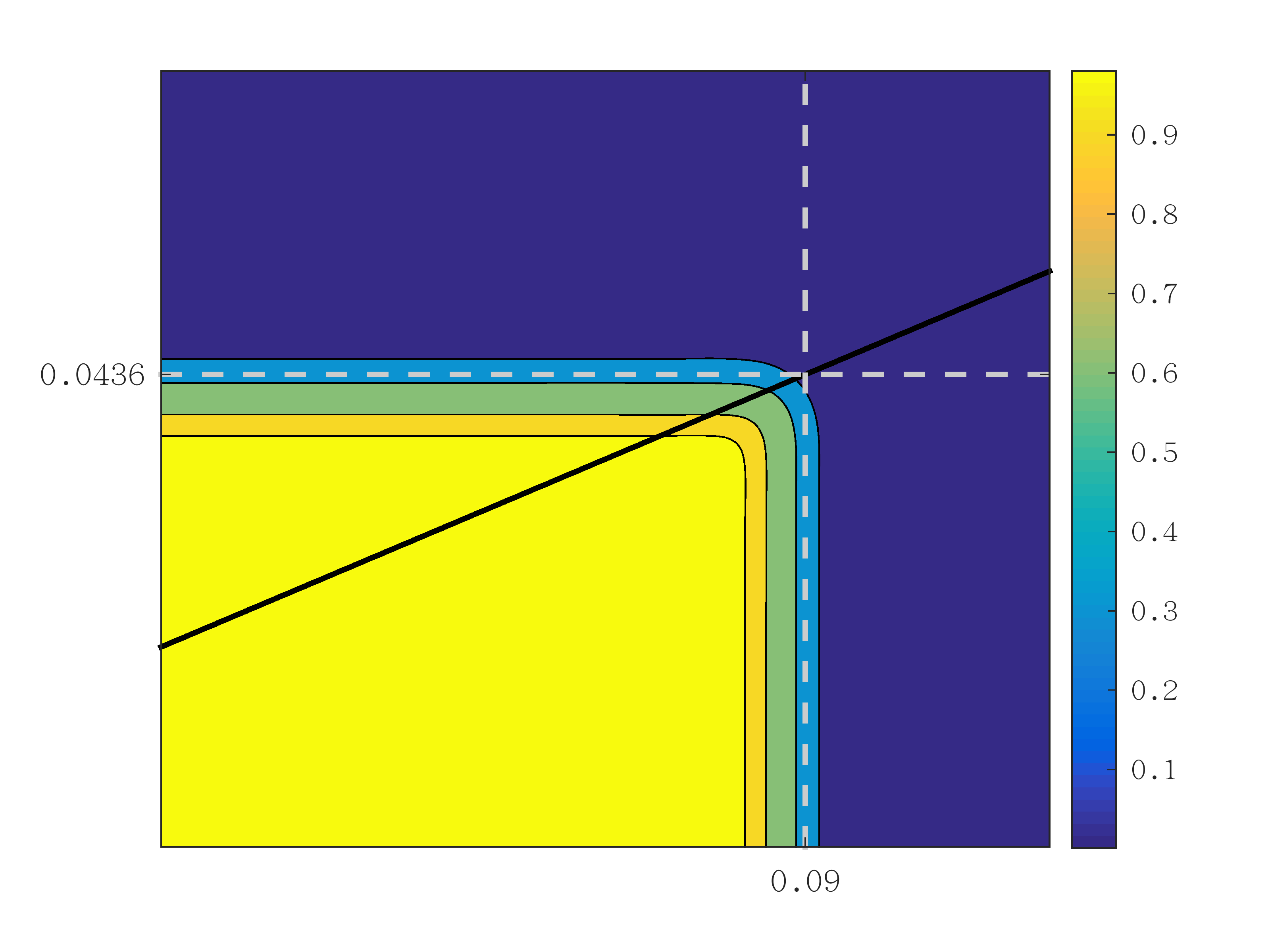}
  }
  \subfloat[$\cos\theta_0 = \dfrac{9}{10}$, $ct_{\text{end}}=0.2$]{
      \includegraphics[width=0.33\textwidth
      ,trim={10mm 8mm 15mm 8mm}, clip
      ]{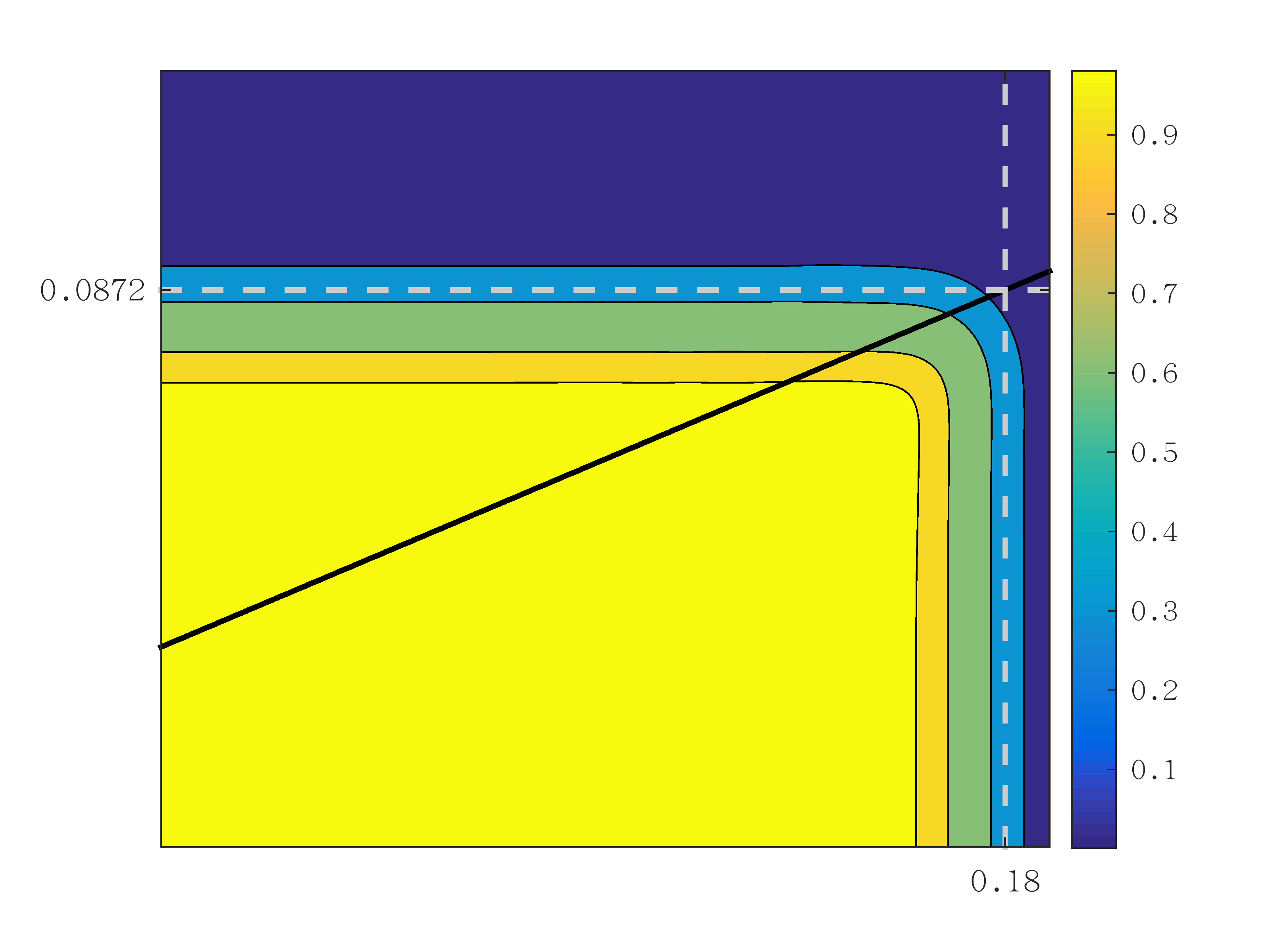}
  }
  \caption{Profile of $E_0$ of the inflow problem with the \HMP{2} model, where the
  directions are different}
  \label{fig:testBeam}
\end{figure}

In order to validate the capability of the 3D \HMPN model to simulate
Dirac delta functions in any direction, we choose some different 
$\bOmega_0$, with $\bOmega_0 = (\cos\theta_0, \sin\theta_0)$, and
$\cos\theta_0=\dfrac{3}{10}$, $\dfrac{6}{10}$, and $\dfrac{9}{10}$.
We simulate this
problem with the 3D \HMPN model till $ct_{\text{end}}=0.05, 0.1$, and
$0.2$, with $N=2$. $N_x=N_y=160$ is applied in this example.
The results of the contours of $\dfrac{E_0}{ac}$ are presented in
\Cref{fig:testBeam}. 
The black line represents the direction of
$\bOmega_0$, i.e. the slope of each line is $\tan\theta_0$, and the
dashed lines ($x=ct_{\text{end}}\cos\theta_0$ and
$y=ct_{\text{end}}\sin\theta_0$) figure out the theoretical results 
of the wave in
$x$-direction and $y$-direction, respectively.
From the
results, one can conclude that the 3D \HMPN model can capture the
wave well. 
Therefore, the approximation of the \HMPN to the
delta function in any direction is satisfying.

\paragraph{Gaussian source problem}
This example is to show the rotational invariance of the \MPN
model. The initial specific intensity is a Gaussian distribution in
space \cite{Frank2012Perturbed}:
\begin{align}
  I_0(\bx; \bOmega) =
  \dfrac{1}{4\pi}\dfrac{ac}{\sqrt{2\pi\theta}}e^{-\frac{\Vert\bx\Vert^2}{2\theta}},
  \quad \theta = \frac{1}{100},\quad  \bx \in (-L,L)\times(-L,L).
\end{align}
Here the computational domain is set by $L=1$, and
$ct_{\text{end}}=0.5$ so that that the energy reaching the boundaries
is negligible, and vacuum boundary conditions are prescribed at all
the boundaries.  The medium is purely scattering with
$\sigma_s=\sigma_t=1$, thus the material coupling term vanishes. We
also set the external source to be zero.

\begin{figure}[htbp]
  \centering 
  \subfloat[$N=2$]{
      \includegraphics[width=0.33\textwidth,
  trim={15mm 8mm 20mm 10mm}, clip]{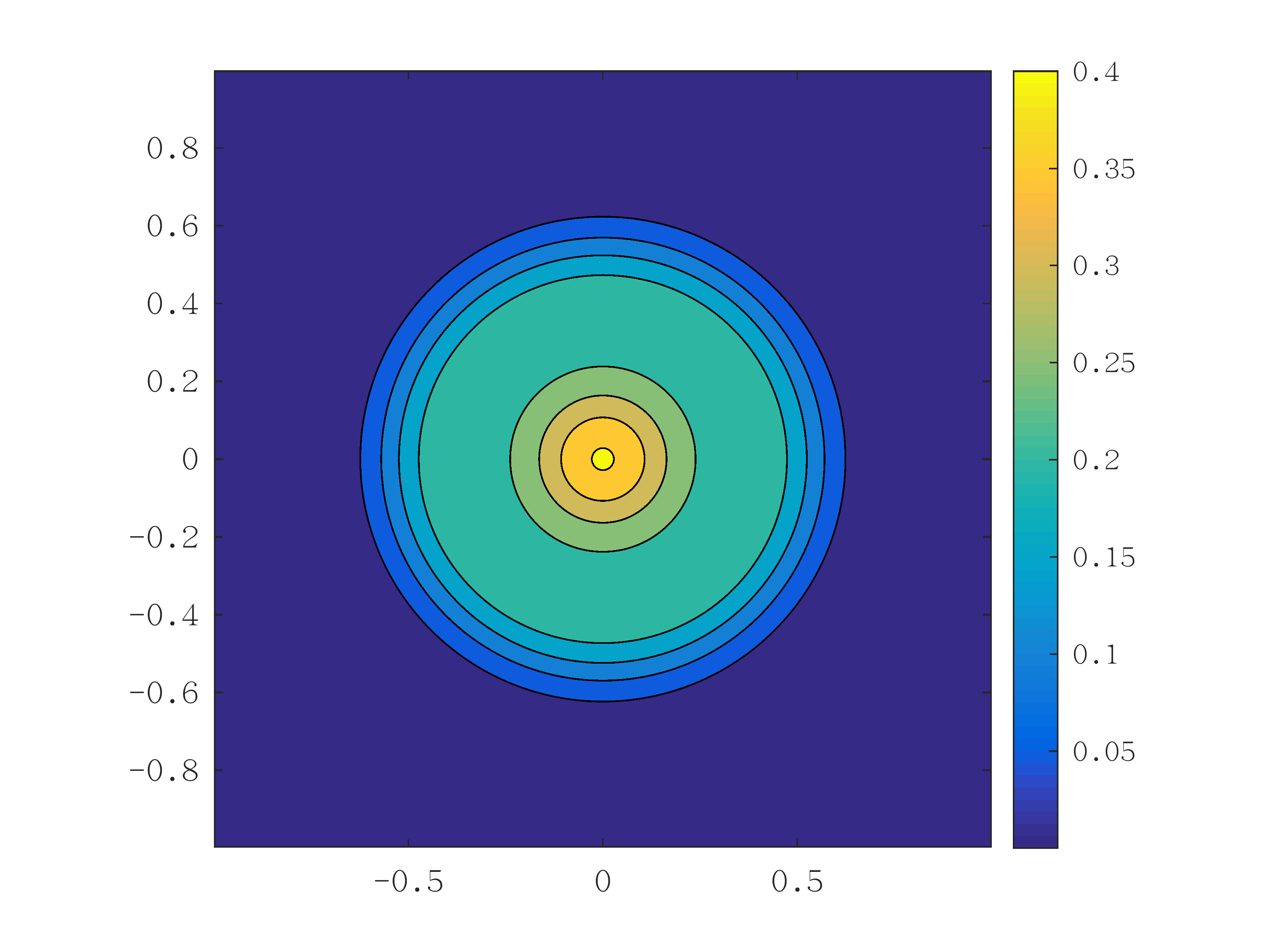}
  }
  \subfloat[$N=3$]{
      \includegraphics[width=0.33\textwidth,
  trim={15mm 8mm 20mm 10mm}, clip]{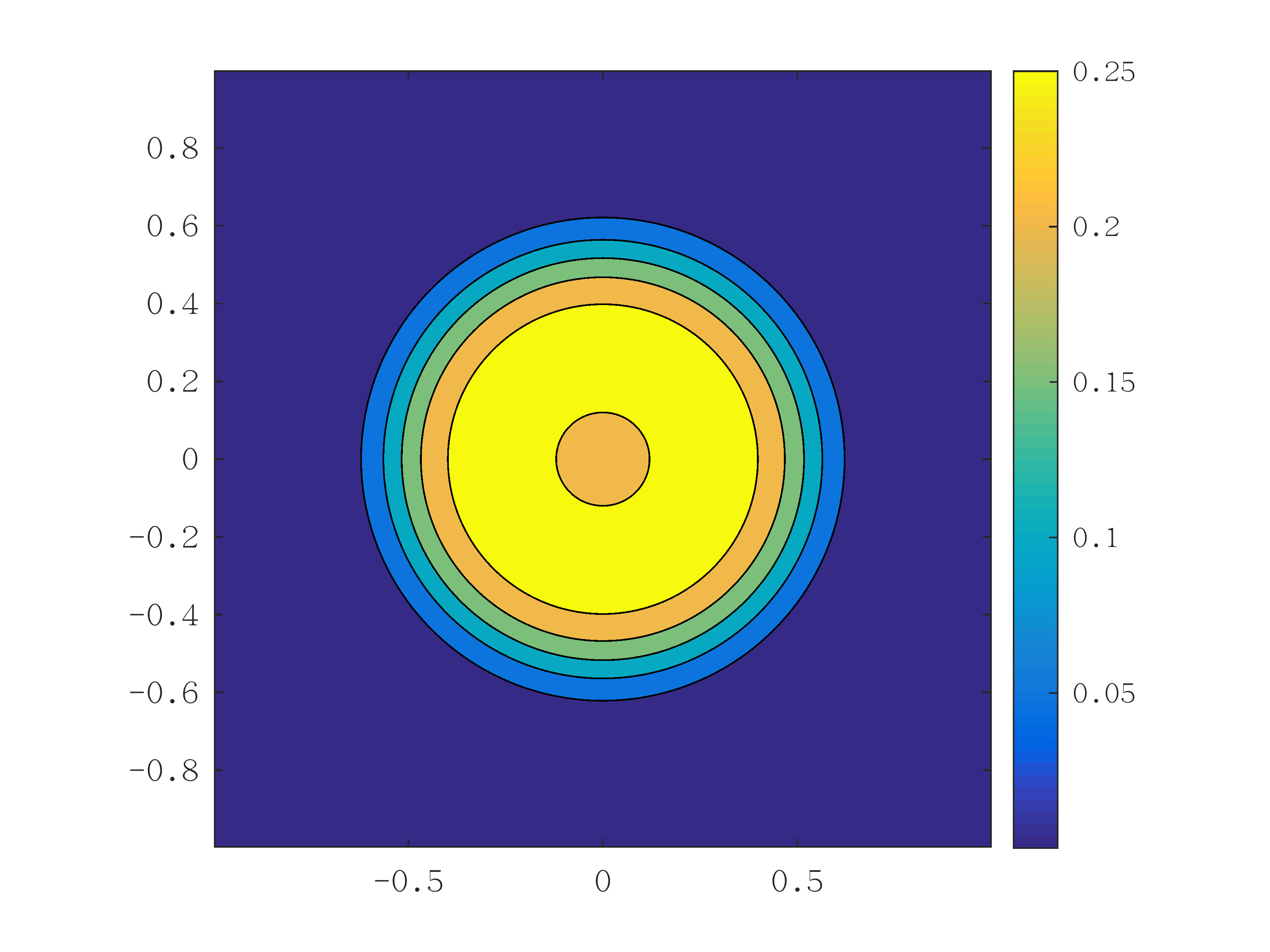}
  }
  \subfloat[$N=4$]{
      \includegraphics[width=0.33\textwidth,
  trim={15mm 8mm 20mm 10mm}, clip]{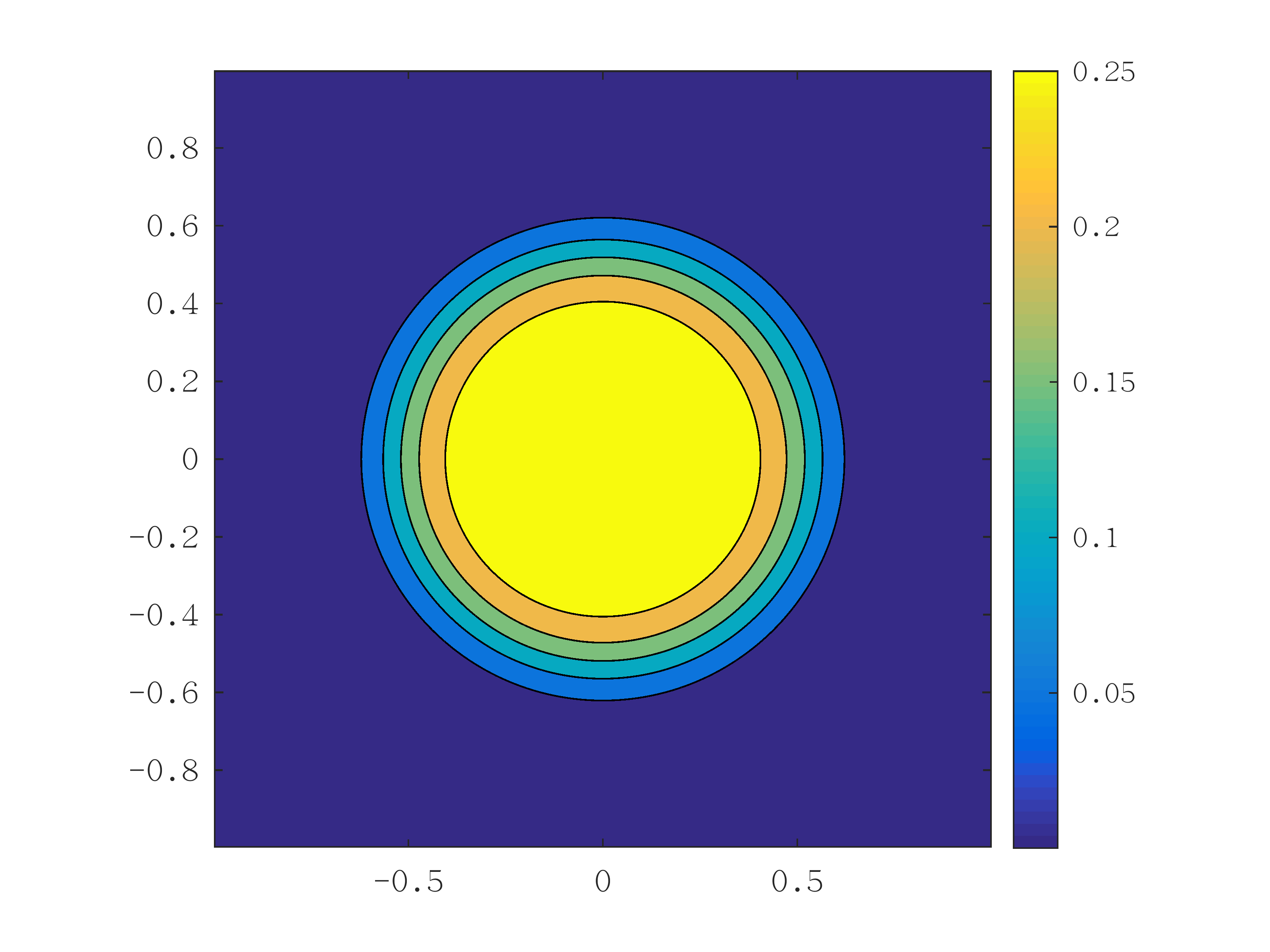}
  }
  \caption{Profile of $E_0$ of the Gassian source problem 
  with the \HMPN model}
  \label{fig:testLS}
\end{figure}

We use $400\times400$ cells to simulate this problem, with the \HMP{2}
model, the \HMP{3} model, and the \HMP{4} model, and the results are
presented in \Cref{fig:testLS}, from which one can conclude the
contours of the $\dfrac{E_0}{ac}$ are circles and there is no ray effect. This
validates the rotational invariance of the 3D \HMPN model.

\paragraph{Bilateral inflow problem}
In the previous numerical example, we validate the capability of the 
3D \HMPN model to simulate Dirac delta function. In this example, we
show that the 3D \HMPN model can also approximate isotropic
distribution function. 

We use the RTE without right hand side \eqref{eq:RTE_norhs}, and 
the initial state is chosen as
\[
  I_0(\bx; \bOmega) =\left\{
\begin{aligned}
  &ac\delta(\bOmega-\bOmega_0),\quad & x<-\dfrac{2}{10} 
  \text{ and } y<-\dfrac{2}{10} ,\\
  &\dfrac{ac}{4\pi},\quad & x>\dfrac{2}{10} \text{ and } y > \dfrac{2}{10},\\
  &\dfrac{10^{-3}}{4\pi}ac,\quad &\text{otherwise},
\end{aligned}
\right.
\]
where $\bOmega_0 = (\frac{\sqrt{2}}{2}, \frac{\sqrt{2}}{2})^T$. 
In the bottom-left region, the distribution function is a Dirac delta 
function, which is extremely anisotropic and hard to approximate 
with the \PN model, which approximates the specific intensity with polynomials. 
Meanwhile, in the upper-right region, the distribution function is a
constant with respect to the velocity direction $\bOmega$, which is 
isotropic. Generally, it is challenging to get a good approximation 
on a Dirac delta function and a constant function. 

The computational domain is $[-1, 1]\times[-1, 1]$, and the infinite
boundary conditions are prescribed at the boundaries.  We simulate this
problem with the \MPN model and the \PN model till $ct_{\text{end}}=0.1$ 
with $N=2$, $3$, and $4$. $N_x=N_y=400$ is applied in this example.

\begin{figure}[htbp]
  \centering 
  \subfloat[\HMP{2} model]{
  \includegraphics[width=0.33\textwidth,trim={15mm 5mm 20mm 10mm}, 
  clip]{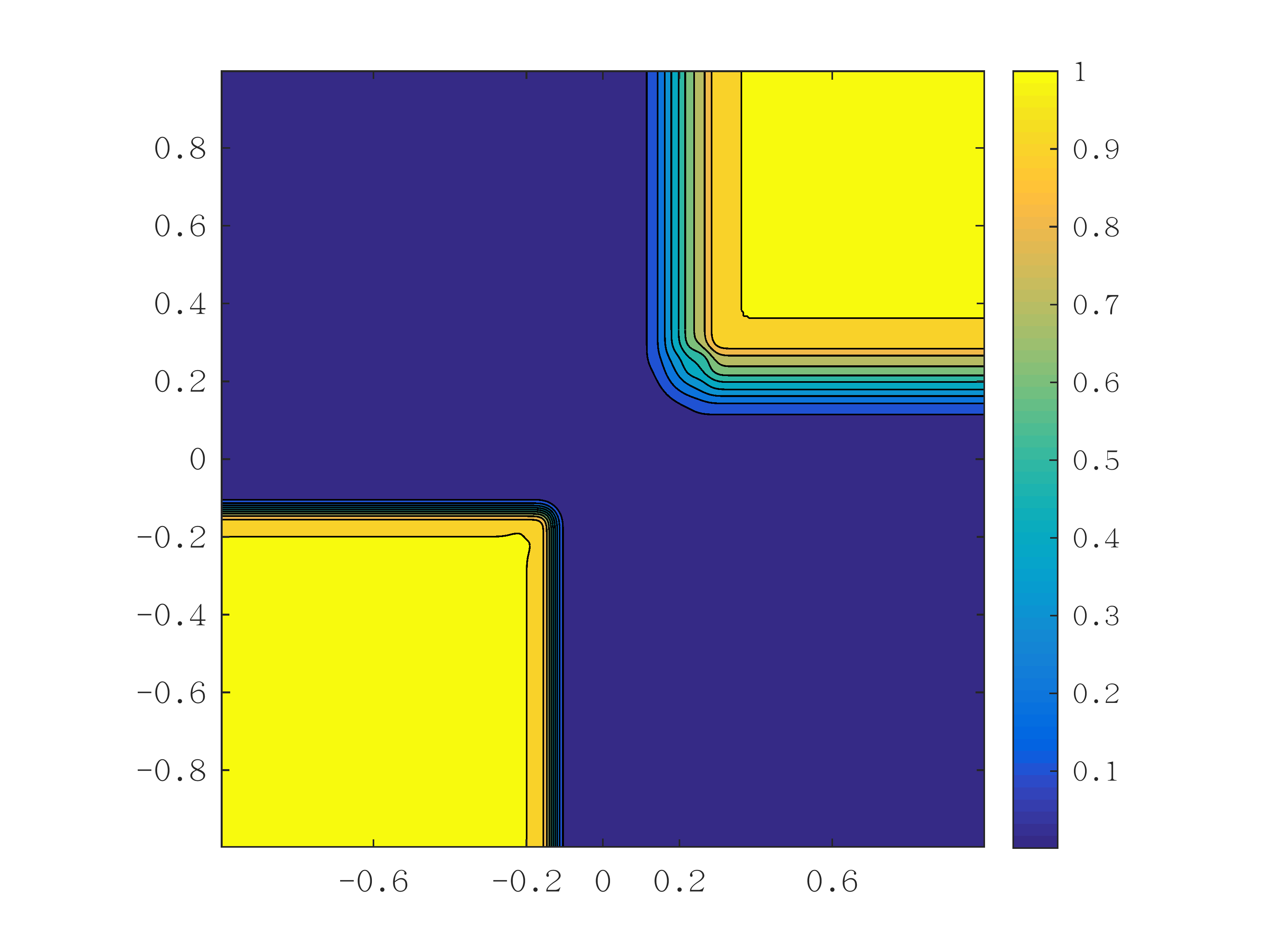}
  }
  \subfloat[$P_2$ model]{
  \includegraphics[width=0.33\textwidth,trim={15mm 5mm 20mm 10mm},
  clip]{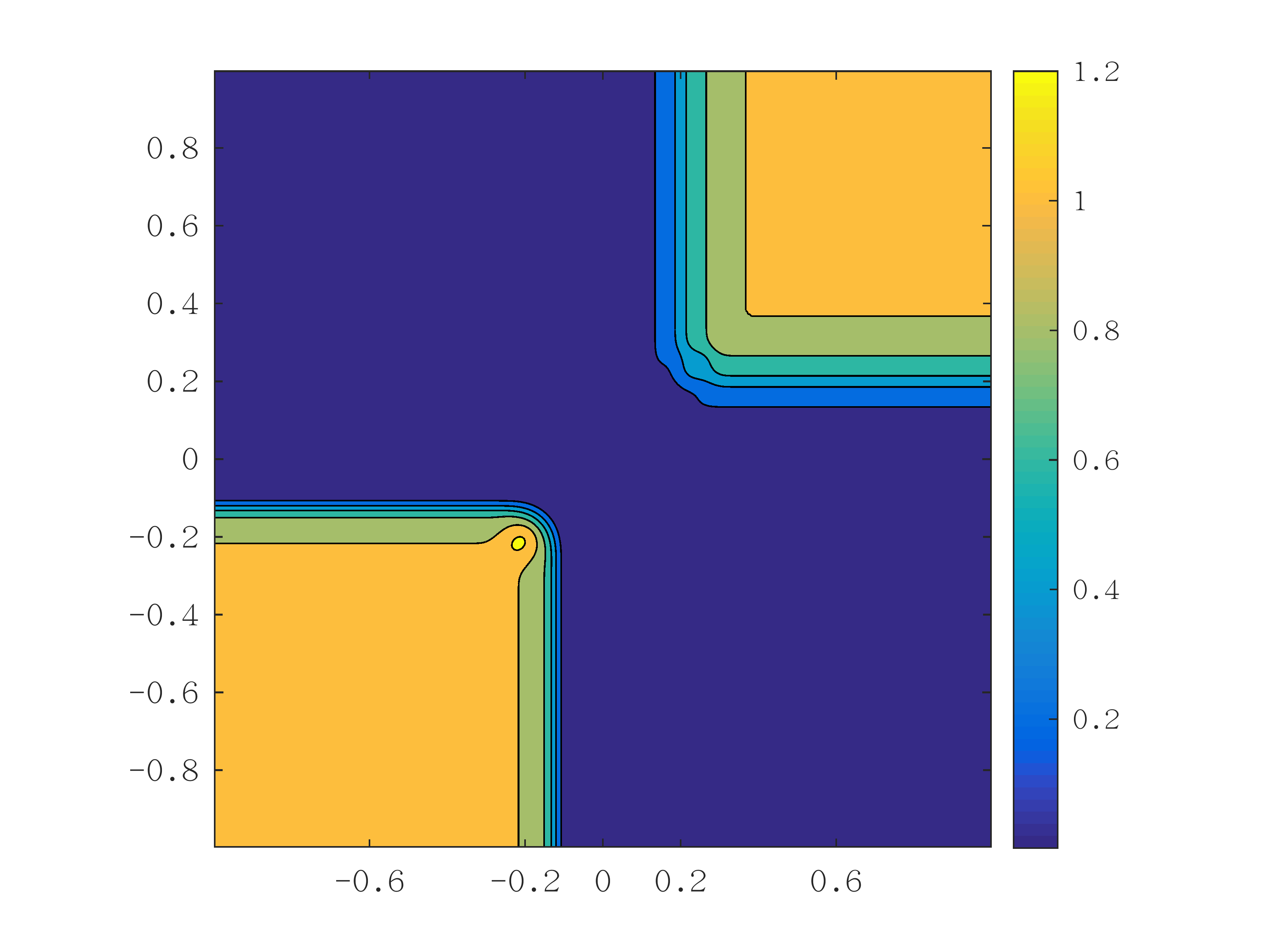}
  }
  \subfloat[$E_0$ when $y=x$]{
      \includegraphics[width=0.33\textwidth]{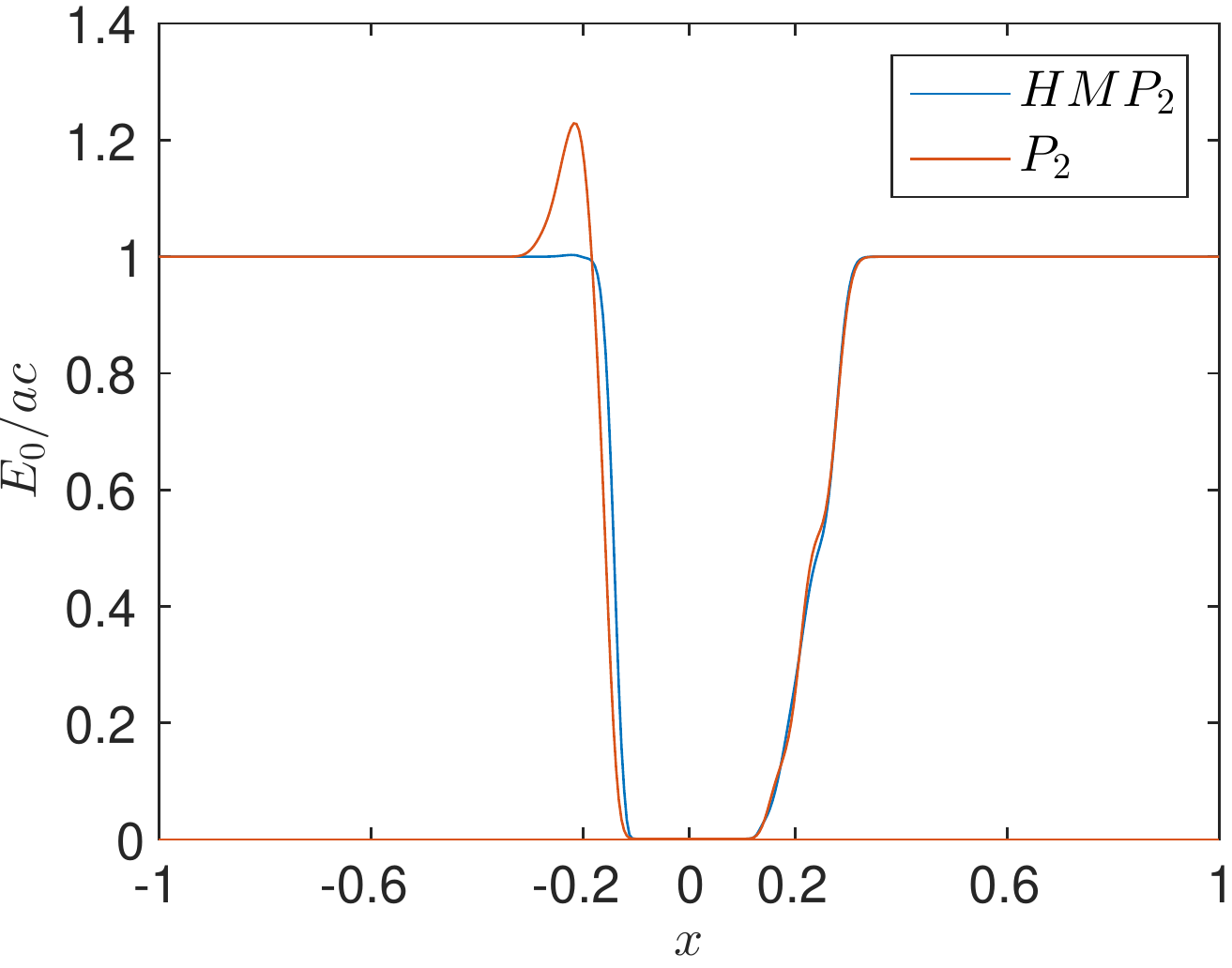}
  }\\
  \subfloat[\HMP{3} model]{
  \includegraphics[width=0.33\textwidth,trim={15mm 5mm 20mm 10mm},
  clip]{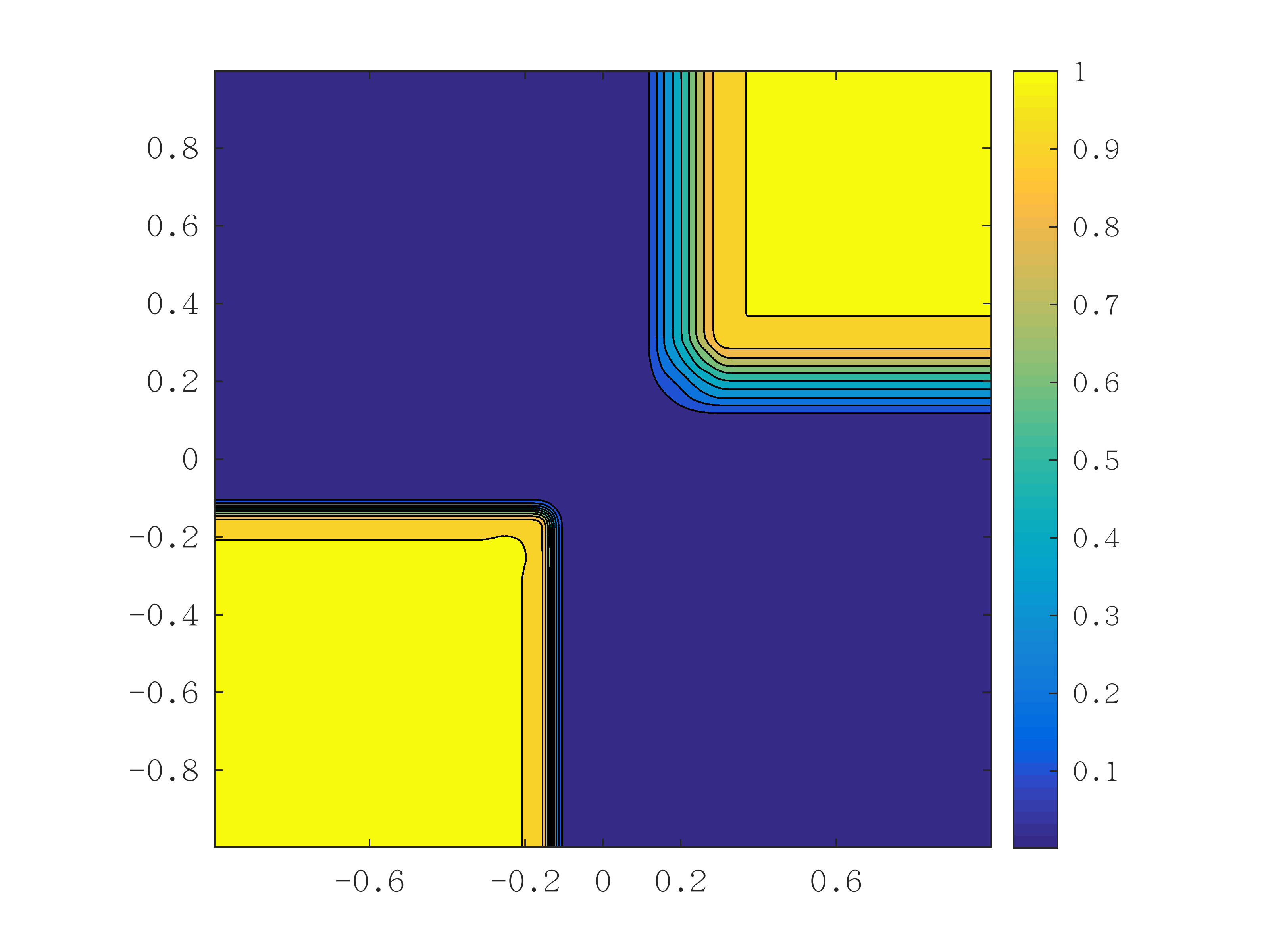}
  }
  \subfloat[$P_3$ model]{
  \includegraphics[width=0.33\textwidth,trim={15mm 5mm 20mm 10mm},
  clip]{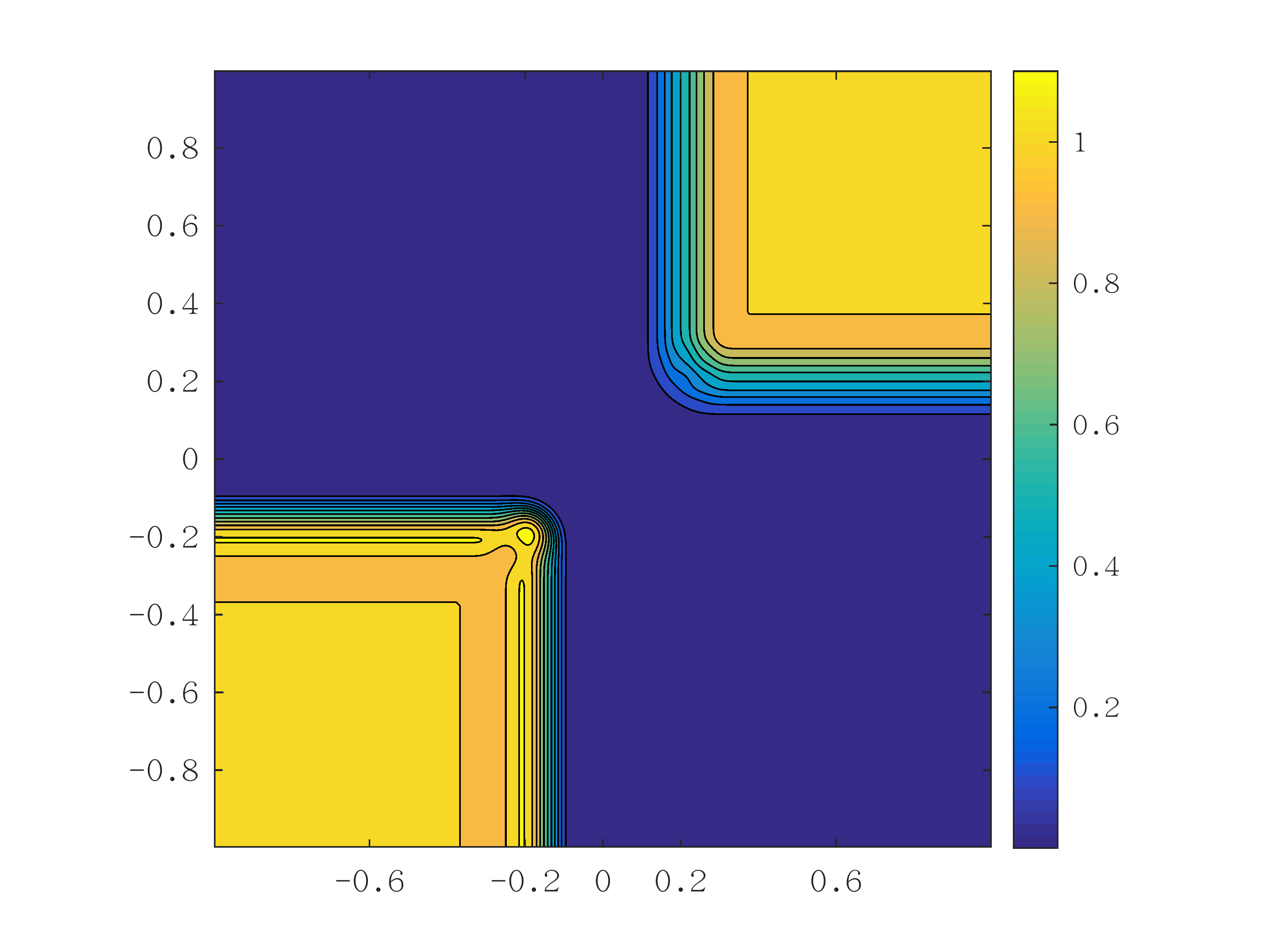}
  }
  \subfloat[$E_0$ when $y=x$]{
      \includegraphics[width=0.33\textwidth]{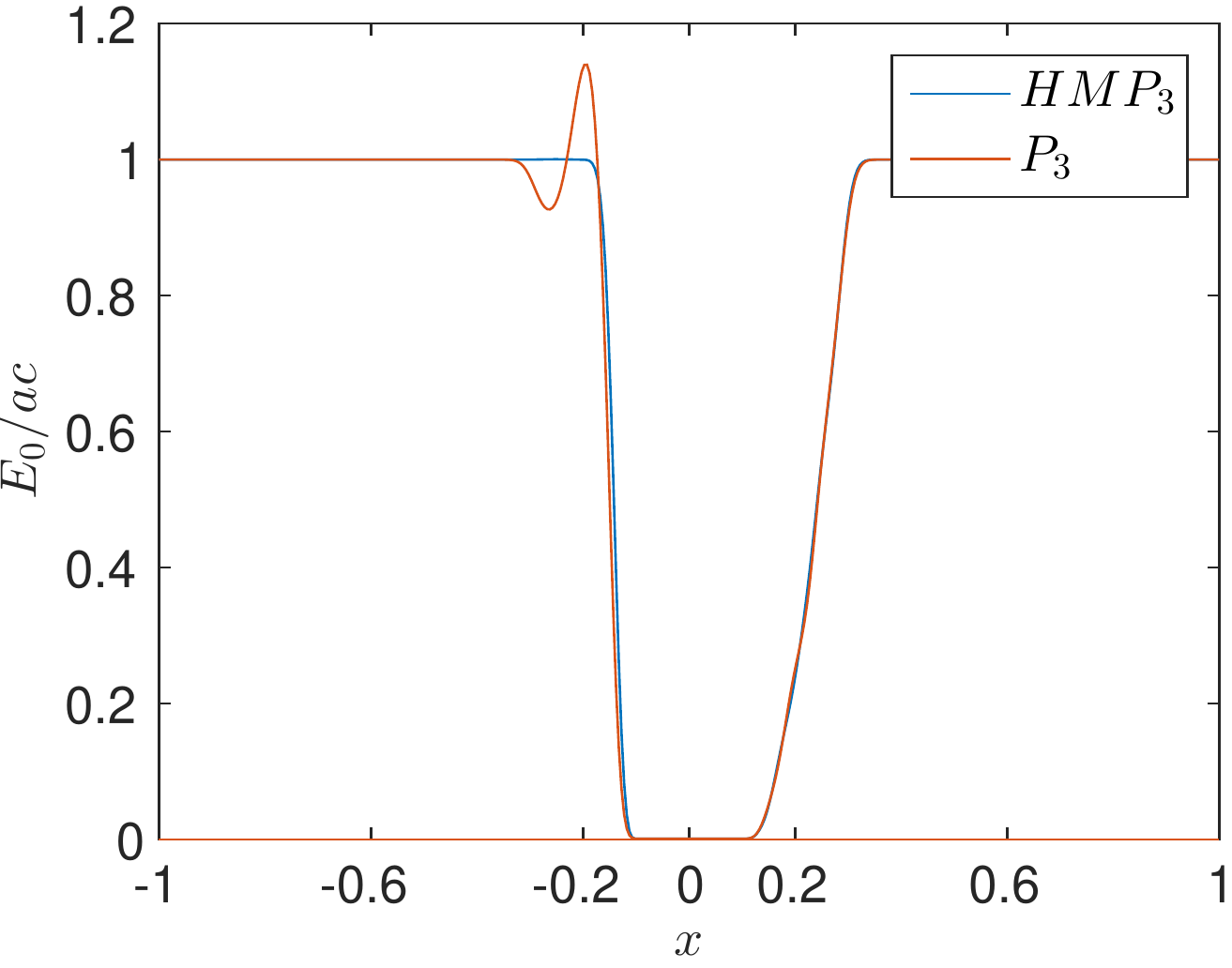}
  }\\
  \subfloat[\HMP{4} model]{
  \includegraphics[width=0.33\textwidth,trim={15mm 5mm 20mm 10mm},
  clip]{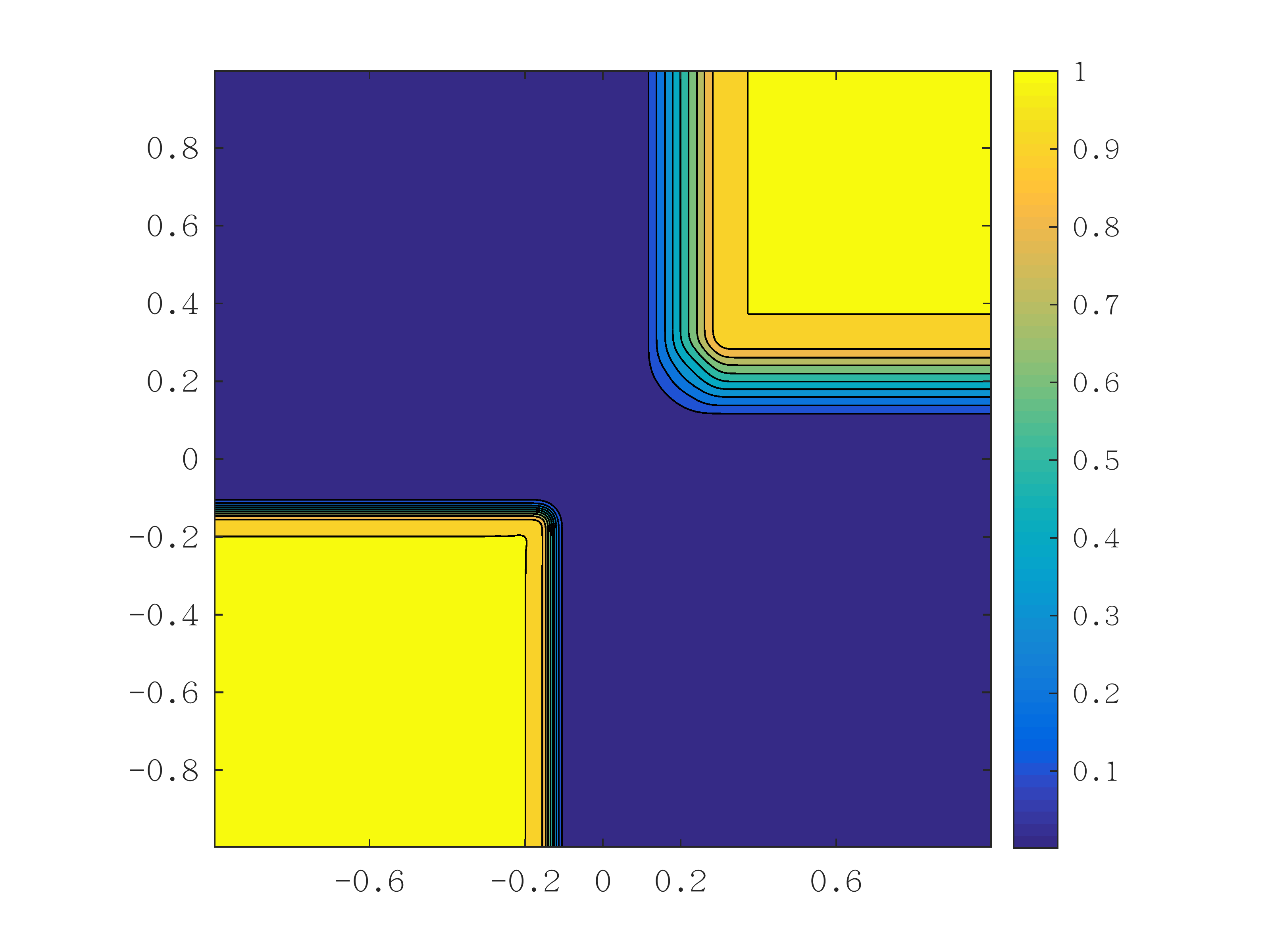}
  }
  \subfloat[$P_4$ model]{
  \includegraphics[width=0.33\textwidth,trim={15mm 5mm 20mm 10mm},
  clip]{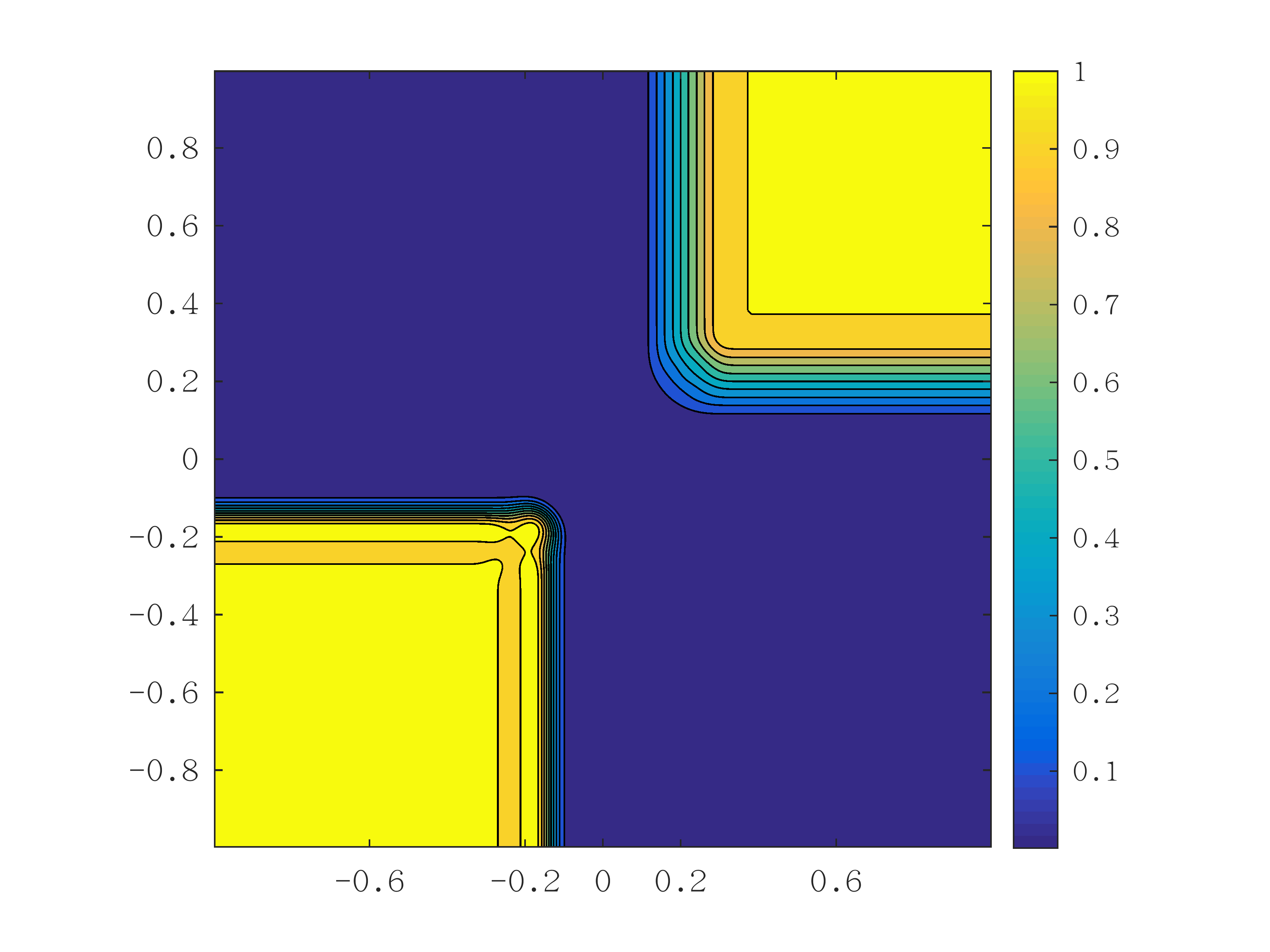}
  }
  \subfloat[$E_0$ when $y=x$]{
      \includegraphics[width=0.33\textwidth]{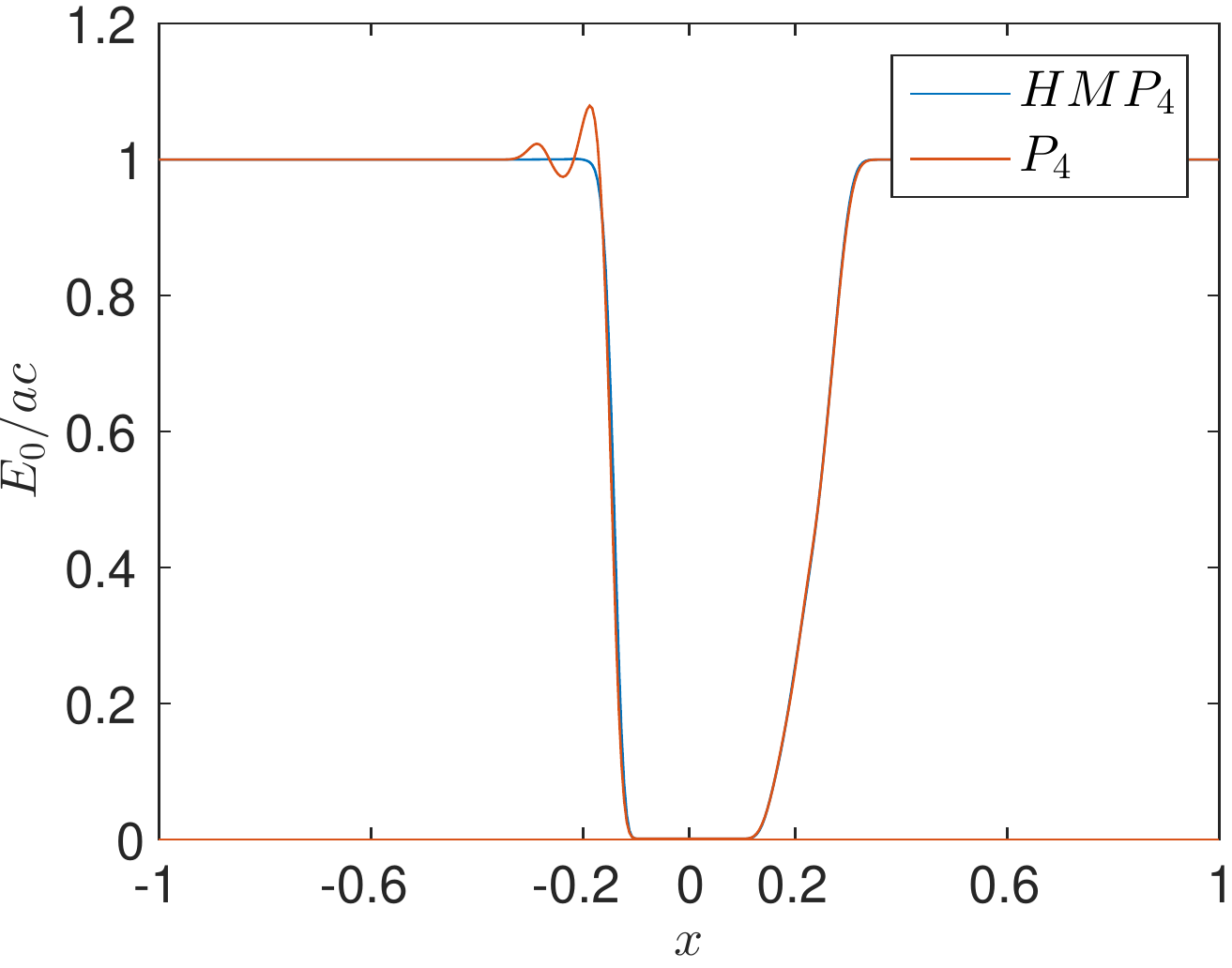}
  }

  \caption{Profile of $E_0$ of the bilateral inflow with the \HMP{2} model
  and the $P_2$ model} 
  \label{fig:testDelta}
\end{figure}
In \cref{fig:testDelta}, we present the $E_0$ of the 3D \HMPN model and
the \PN model for $N=2$, $3$, and $4$, 
and additionally, we present the $\dfrac{E_0}{ac}$ of these two
models along the line $y=x$. According to the results in
\cref{fig:testDelta}, we can conclude that in the upper-right region,
for this constant function, the \HMPN model and the \PN model gets
similar results, i.e. these two models can approximate this kind of
isotropic function well. However, in the bottom-left region,
for this Dirac delta function, the \HMPN model gets a good
approximation, but there are unphysical oscillations in the results of 
the \PN model. In particular, due to these unphysical oscillations,
along the line $y=x$, $\dfrac{E_0}{ac}$ of
the \PN model can be even greater than $1$.
Therefore, one can conclude that the \HMPN model can not only get a
good approximate on an isotropic distribution function, but also
approximate the Dirac delta function well. 

\paragraph{Lattice problem}
The lattice problem is a checkerboard of highly scattering and highly
absorbing regions loosely based on a small part of a lattice core.
The computational domain is $[0,7]\times[0,7]$, divided into 49 grids
in \Cref{fig:Lattice_problem}. In the red grids and the black grid,
the absorbing and scattering coefficients are 0 and 1 respectively. In
the white grids, the absorption and scattering are 10 and 0,
respectively. The external source $s$ is set as $ac$ in the black grid,
and 0 in other grids.  The initial state is set as $I=10^{-8}ac$, and
vacuum boundary condition are prescribed on all boundaries.
\begin{figure}[htbp]
  \centering
\begin{tikzpicture}
  \draw[step=1, help lines] (0,0) grid (7,7);
  \foreach \i in {1,3,5}
  \foreach \j in {1,3,5}
  \filldraw[red] (\i, \j) -- (\i, \j+1) -- (\i+1, \j+1) -- (\i+1,\j)
  -- cycle ;
  \foreach \i in {2,4}
  \foreach \j in {2,4}
  \filldraw[red] (\i, \j) -- (\i, \j+1) -- (\i+1, \j+1) -- (\i+1,\j)
  -- cycle ;
  
  \filldraw[black] (3, 3) -- (3, 3+1) -- (3+1, 3+1) -- (3+1,3)
  -- cycle ;
  \filldraw[white] (3, 5) -- (3, 5+1) -- (3+1, 5+1) -- (3+1,5)
  -- cycle ;
  \draw[step=1, help lines] (0,0) grid (7,7);
\end{tikzpicture}
\caption{Computational domain of the lattice problem}
\label{fig:Lattice_problem}
\end{figure}
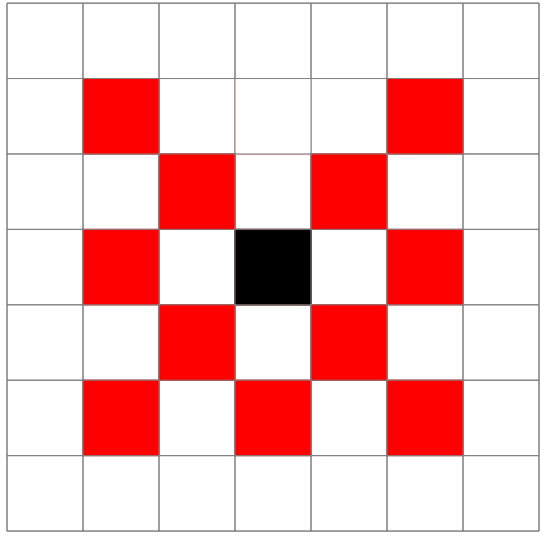

\begin{figure}[htbp]
  \centering 
  \subfloat[$N=2$]{
      \includegraphics[width=0.33\textwidth,
  trim={13mm 3mm 15mm 10mm}, clip]{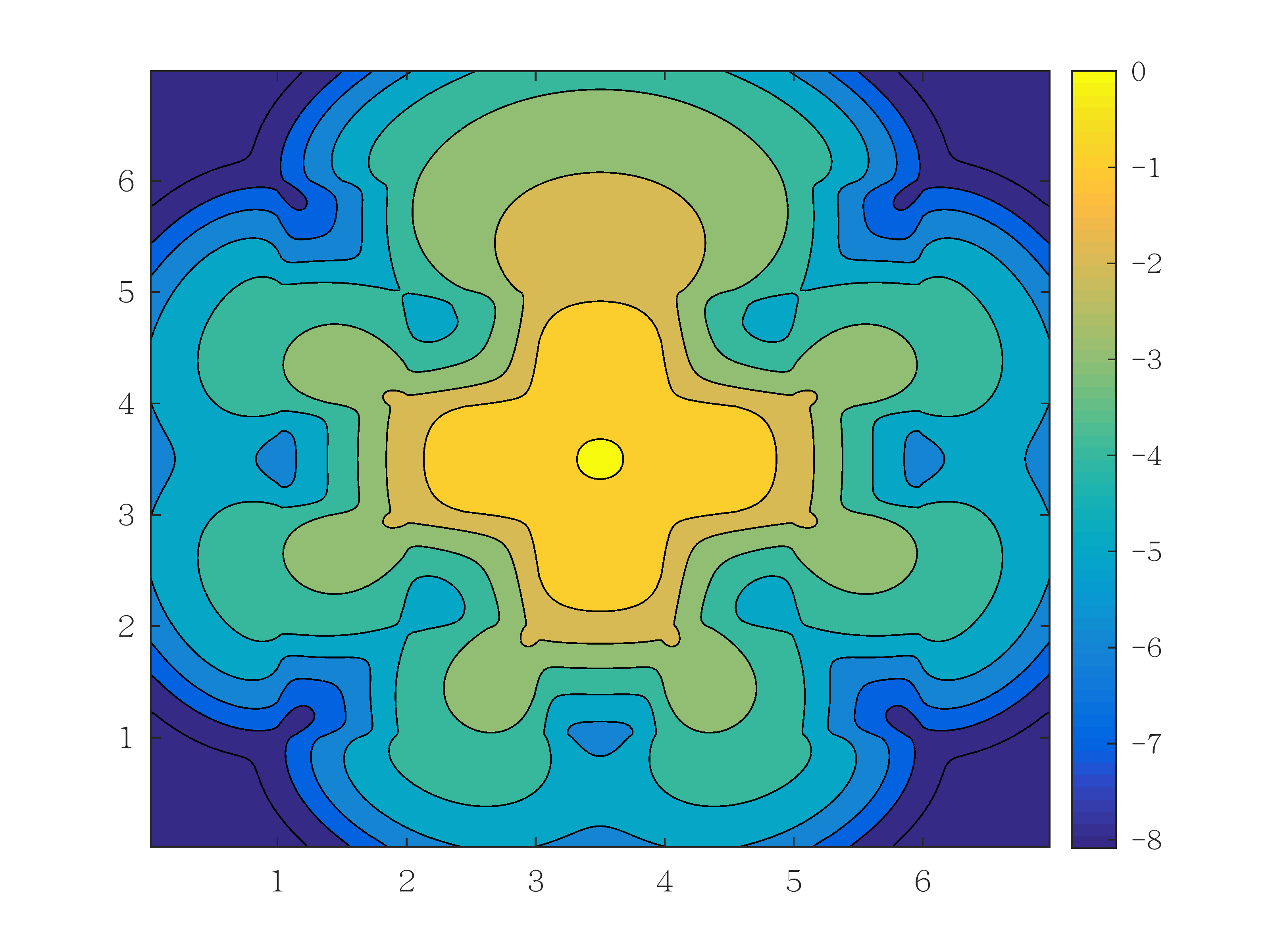}
  }
  \subfloat[$N=4$]{
      \includegraphics[width=0.33\textwidth,
  trim={13mm 3mm 15mm 10mm}, clip]{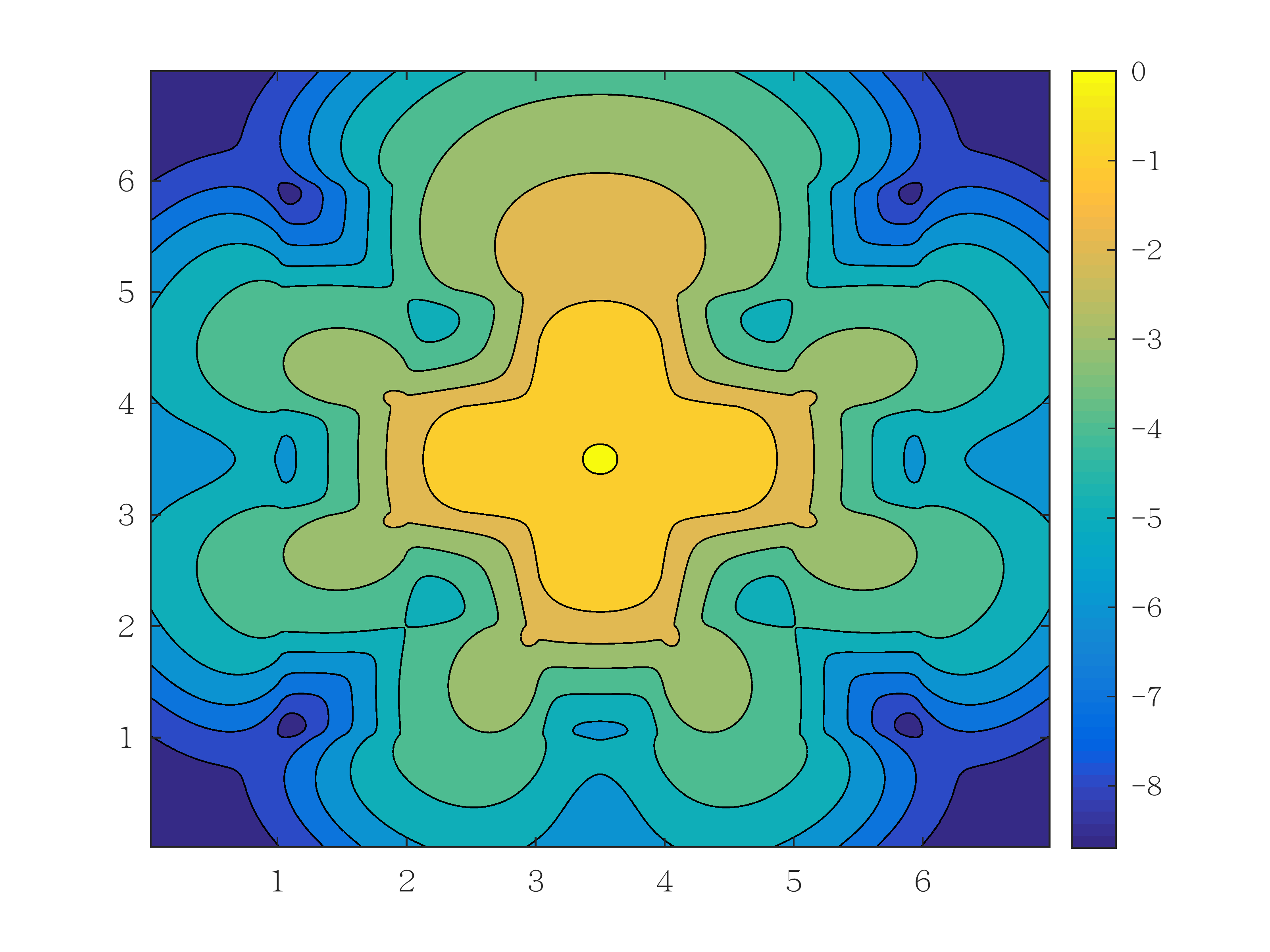}
  }
  \subfloat[$N=6$]{
      \includegraphics[width=0.33\textwidth,
  trim={13mm 3mm 15mm 10mm}, clip]{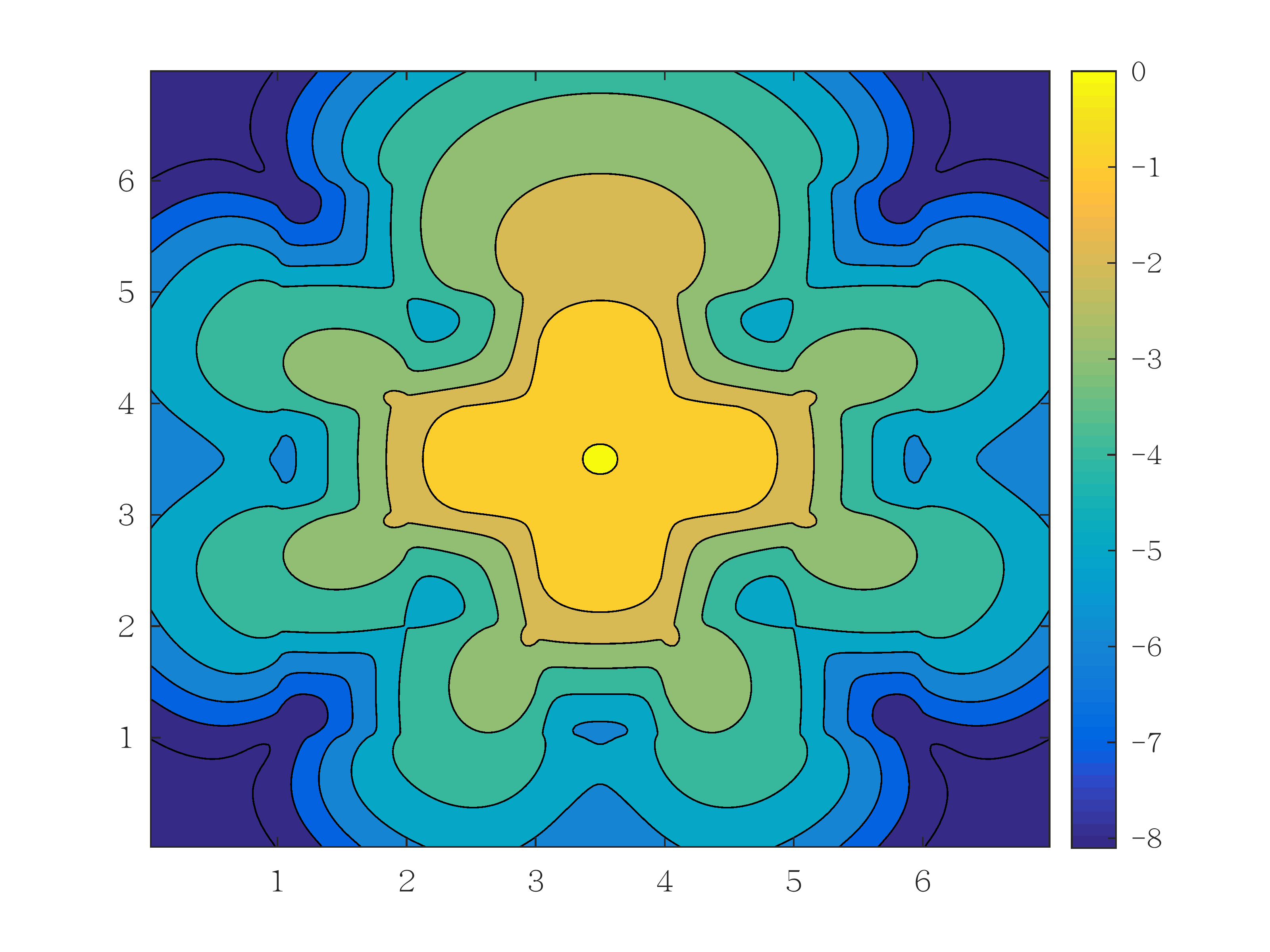}
  }
  \caption{Profile of $E_0$ of the lattice problem 
  with the \HMPN model}
  \label{fig:Lattice_problem_result}
\end{figure}
We simulate this problem with the \HMPN model with $N=2$, $4$, and $6$
till $ct_{\text{end}}=3.2$, the number of grids is $200\times 200$, 
and the results of $\log_{10}\dfrac{E_0}{ac}$ are presented in \Cref{fig:Lattice_problem_result}. 
According to the results, the \HMPN model gives a satisfying result, 
and the result of $E_0$ is always
positive in this problem.


\section{Conclusion} \label{sec:conclusion}
The 3D \HMPN model was derived as a reduced nonlinear model for RTE in
3D space. The model is hopeful to capture both very singular
specific intensity as Dirac delta function and very regular specific
intensity as constant function. The model has some mathematical
advantages, including global hyperbolicity, rotational invariance,
physical wave speeds, spectral accuracy, and correct higher-order
Eddington approximation. We validated the new model by some
preliminary numerical results. In the following, we will try to apply
the model to some practical problems.

\section*{Acknowledgements}
The authors are partially supported by Science Challenge Project,
No. TZ2016002, the CAEP foundation (No. CX20200026),
and the National Natural Science Foundation of China
(Grant No. 91630310 and 11421110001, 11421101).


\bibliographystyle{abbrv}
\bibliography{../../article,../references}
\end{document}